\documentclass[onecolumn,journal,12pt]{IEEEtran}
\usepackage[noadjust]{cite}

\ifCLASSINFOpdf
\else
\fi
\usepackage{amssymb}
\usepackage{dsfont}
\usepackage{array}
\usepackage{tabularx}
\usepackage{makecell}
\usepackage{multirow}
\newcolumntype{C}[1]{>{\centering\arraybackslash}p{#1}}
\newcolumntype{Y}{>{\raggedright\arraybackslash}X}
\usepackage{float}
\usepackage{graphicx}

\newtheorem{lemma}{Lemma}

\newtheorem{proposition}{Proposition}
\newtheorem{theorem}{Theorem}
\newtheorem{remark}{Remark}
\newtheorem{corollary}{Corollary}

\usepackage{xcolor}
\usepackage{mathrsfs}
\usepackage{mathtools}
\usepackage{enumitem}
\usepackage{caption}
\usepackage[ colorlinks = true,              linkcolor = blue, urlcolor  =
blue,              citecolor = red,              anchorcolor = green,
]{hyperref}

\newcommand{\ind}{\perp \!\!\! \perp}

\usepackage [autostyle, english = american]{csquotes}
\MakeOuterQuote{"}

\usepackage{amsmath}

\begin{document}
%

\title{Joint Source-Channel Coding of Gaussian Sources over Block Erasure Channels: Nonasymptotic Bounds and Channel-Uniform Normal Approximations}
%
%
%

\author{Adeel Mahmood\\Nokia Bell Labs\\Murray Hill, NJ, USA}

\maketitle


\begin{abstract}
We study finite-blocklength lossy transmission of a Gaussian memoryless
source over memoryless, possibly nonstationary block erasure channels
under an excess mean-squared distortion criterion. We derive computable
nonasymptotic achievability and converse bounds and establish matching
third-order, channel-uniform normal approximations. Our achievability bound exactly evaluates the ensemble-average excess-distortion probability of a specified random coding scheme and improves corresponding specializations of known general one-shot bounds. Our converse conditions on the erasure pattern and source energy and combines spherical-cap and volume bounds, with the spherical-cap term providing the geometric prefactor needed for the matching third-order term. In our channel-uniform normal approximation, we show that for fixed distortion ratio,
target excess-distortion probability, and block size, the sufficient and necessary information-balance conditions have the same dispersion and $\frac{1}{2}\log k$ terms, where $k$ is the source blocklength, and differ only by bounded remainders that are uniform over the channel blocklength and erasure-probability profile. Numerical evaluations compare the nonasymptotic JSCC bounds and their common third-order normal approximation with an optimized symmetrized SSCC achievability benchmark. 
\end{abstract}
 
\section{Introduction}

Finite-blocklength information theory has devoted substantial effort to characterizing communication systems that cannot operate with very long blocklengths because of energy, latency and/or complexity constraints. The normal approximation of Polyanskiy, Poor and Verd\'u \cite{5452208} was shown to accurately describe channel coding limits at a fixed nonzero error probability $\epsilon \in [10^{-6}, 10^{-1}]$ and blocklengths on the order of a few hundred channel uses, and in representative examples even around $n=150$. Subsequent work \cite{6802432} continued to develop accurate approximations for a variety of channel models for similarly small blocklengths  that are relevant for the paradigm of ultra-reliable low-latency communication \cite{durisi2016toward,lancho2020saddlepoint}. In joint source--channel coding (JSCC), Kostina and Verd\'u \cite{kostina_JSCC} derived refined asymptotic bounds based on the normal approximation that were shown to be reasonably accurate at a fixed excess-distortion probability $\epsilon \approx 10^{-2}$ and source blocklength $k$ and channel blocklength $n$ on the order of hundreds.  

In this paper, we derive more accurate JSCC bounds for the special case of a Gaussian memoryless source $\mathcal{N}(0, \sigma^2)$ and nonstationary block erasure channels (BLEC) with erasure probabilities $p_1, \ldots, p_n$. Our starting achievability and converse results, Theorems \ref{thm_achievability} and \ref{convo2}, are nonasymptotic bounds that hold for all $n \geq 1$ and $k \geq 1$. The random coding achievability scheme underlying our nonasymptotic result is tailored for block erasure channels and a Gaussian memoryless source (GMS), with straightforward adaptation to general source laws also given. Theorem \ref{thm_achievability} derives an exact expression of the ensemble-average excess-distortion probability of the scheme. The resulting achievability bound improves upon mere specializations to the GMS-BLEC pair of the previous general 
one-shot JSCC achievability bounds \cite[Theorem~8]{kostina_JSCC} and \cite[Theorem~4]{alternative_oneshot_JSCC} evaluated with the same channel input and source reproduction distributions. Also, like \cite[Theorem~4]{alternative_oneshot_JSCC} and unlike \cite[Theorem~8]{kostina_JSCC}, our nonasymptotic achievability bound requires no optimization over an auxiliary scalar parameter.

The converse in Theorem~\ref{convo2} exploits the geometry of the
Gaussian source after conditioning on the erasure pattern. If exactly $j$
blocks are unerased, a deterministic decoder can produce at most $2^{mj}$ possible
reproduction sequences. Conditioning further on the source energy
allows us to upper bound the fraction of the source sphere covered by their
distortion balls. We combine this spherical-cap bound with a
complementary volume bound, taking the stronger bound separately for
each $j$ before averaging. The resulting converse is at least as
strong as the Lebesgue-measure specialization of the list-code converse
in \cite[Theorem~5]{kostina_JSCC} with the auxiliary output law
specified in~\eqref{2pyi}. Its spherical-cap component supplies the
third-order term that matches our achievability result in our normal-approximation bounds. 

Although nonasymptotic bounds are valid at every blocklength, they have to be numerically evaluated and offer less immediate design insight than closed-form normal approximations. In the normal-approximation results, both sufficient and necessary
conditions for reliable joint source--channel coding can be expressed
through an \emph{information-balance condition}, which compares the
information supplied by the channel with the information required to
represent the source at the target distortion. In its usual
second-order form, this balance condition has the form \cite{kostina_JSCC}
\begin{equation}
    nC - kR(d) \geq \sqrt{nV_c + k V(d)} \, Q^{-1}(\epsilon)  +\text{higher-order terms}, 
\end{equation}
where $C$ and $V_c$ are the
channel capacity and channel dispersion, respectively, $R(d)$ and
$V(d)$ are the source rate-distortion and rate-dispersion functions,
and $Q$ is the standard normal complementary CDF. Intuitively, \(nC\) represents the first-order information that the channel can deliver, while \(kR(d)\) represents the first-order information required to describe the source to distortion \(d\); the dispersions \(V_c\) and \(V(d)\) quantify the corresponding fluctuations around these first-order quantities.

For the GMS--BLEC pair, our Theorems~\ref{asymp_thm_achievability}
and~\ref{thm_semiasympconverse} give sufficient and necessary
information-balance conditions with the same first-order, dispersion,
and $\frac12\log k$ terms, leaving only bounded remainder terms.
For fixed distortion ratio $\delta = d/\sigma^2$ , target excess-distortion probability $\epsilon$, and
block size $m$, the constant remainder terms and source-blocklength thresholds for the information-balance conditions are
uniform over every channel blocklength $n$ and every erasure profile
$p_1,\ldots,p_n\in[0,1]$. We call these channel-uniform normal
approximations. They retain the familiar source and channel dispersion
terms without requiring a separate large-$n$ approximation of the
channel. Previously, a similar channel-uniform converse result \cite[Theorem 3]{adeel_unknown_channels} was proved by Mahmood--Viswanathan–Du, which applies to an arbitrary memoryless source transmission over BLEC, although our Theorem \ref{thm_semiasympconverse} improves its $-\frac12\log k-O(1)$ remainder to $\frac12\log k-O(1)$ for the specific case of a GMS. This is an improvement of
$\log k+O(1)$ in the necessary information-balance condition and matches the logarithmic
term in our achievability condition.

Channel-uniform results are particularly relevant for packetized communication. In our block erasure channel model, one channel use represents the successful delivery or erasure of an $m$-bit block, which may correspond to a packet; although the physical-layer transmission of each packet can involve many channel symbols, the number $n$ of independent packet-delivery events may be small. This model is closer to real network behavior, where the channel impairment is represented as block
erasures instead of bit errors. 
At the same time, an image, video frame, or learned latent representation can contain hundreds or thousands of source samples, making an approximation based on large $k$ more plausible than one requiring large $n$. Recent semantic multimedia and video transmission systems divide encoded content into a small number of blocks---for example, eight or sixteen blocks in representative designs---and reconstruct the source from the subset that is successfully received \cite{homapaper2, nargis}. These examples illustrate a regime in which the source dimension is large, while the packet-level channel blocklength remains too small for a conventional channel-only central-limit approximation.

We focus on the GMS--BLEC pair because it provides both a canonical source model and a practically meaningful packet-loss abstraction. The Gaussian memoryless source (GMS) under mean-square error is a standard benchmark for continuous-valued lossy compression and JSCC and can also serve as an approximate model for decorrelated transform or latent coefficients \cite{kostina_SC,cheng2020}.
At the channel side, the block-erasure model captures systems in which error detection causes a packet or encoded block to be either delivered intact or declared unavailable. Recent semantic communication architectures \cite{homapaper2,nargis} explicitly train their encoders and decoders over uniform and nonuniform (nonstationary) block-erasure channels to handle packet loss, bandwidth limitations, congestion-aware packet dropping, and unequal error protection \cite{adeelUEP}. 

In addition to being channel-uniform, our results sharpen the third-order bounds obtained by directly specializing the general normal approximation of Kostina and Verd\'u \cite[Theorem~10]{kostina_JSCC} to the GMS--BLEC pair, in the stationary setting where both results apply. This improvement partly comes from specializing nonasymptotic bounds to a source--channel pair before carrying out their asymptotic analysis. Such improvements have already been observed in the pair-specific treatments of binary memoryless sources over binary symmetric channels (BMS-BSC) and Gaussian memoryless sources over additive white Gaussian noise channels (GMS-AWGN) \cite[Theorems~15 and~19]{kostina_JSCC}. However, these results in \cite{kostina_JSCC} leave a gap between an $O(1)$ third-order converse term and a $\log k+\log\log k+O(1)$ third-order achievability term. In contrast, our achievability and converse results for the GMS--BLEC pair yield matching $\frac12\log k$ terms with $O(1)$ bounded
remainders, and they also apply to nonstationary channel laws. In our
achievability analysis, a parameter-free one-shot bound and the
logistic random-threshold representation
\cite[Eq.~(31)]{adeel_unknown_channels} avoid additional logarithmic
losses. In our converse analysis, conditioning on the source energy
retains an angular covering constraint that a volume-only bound
discards, while an exponential random-threshold representation \eqref{exportrep}
preserves the resulting refinement without introducing an additional
logarithmic penalty. In both analyses, source-induced Gaussian
smoothing of the channel information then yields an
$O(k^{-1/2})$ normal-approximation error uniformly over the channel
blocklength and erasure profile. Note that Mahmood–Viswanathan–Du used the
logistic random-threshold representation
\cite[Eq.~(31)]{adeel_unknown_channels} 
to also yield a third-order improvement for arbitrary memoryless stationary sources and channels, although their result \cite[Theorem 1]{adeel_unknown_channels} applies to the more general mismatched setting with nonstationary channel laws.

Lastly, for comparison with separate source--channel coding (SSCC),
Proposition~\ref{prop:sscc-achievability} derives a nonasymptotic
achievable benchmark for SSCC with a symmetrized fixed-length digital interface.
The source- and channel-coding results are stated separately in
Propositions~\ref{prop:sscc-source-achievability}
and~\ref{prop:sscc-channel-achievability}, using spherical
reproduction codebooks for the GMS and uniformly generated channel
codebooks for the BLEC. An independent shared permutation permits
the uniform-message channel-coding result to be used without assuming
equiprobable source indices. Combining the two component results
yields the exact ensemble-average excess-distortion probability of
the specified SSCC construction. The numerical evaluation
compares this SSCC benchmark, optimized over the interface size,
with the nonasymptotic JSCC bounds and their common third-order
normal approximation.

\section{Preliminaries}

All logs are natural logs and $\exp(x)$ is the natural exponent, unless otherwise noted. We use $\|\cdot\|$ for the Euclidean norm. We use $F_{\chi_k^2}$ to denote the CDF of the $\chi^2$ distribution with $k$ degrees of freedom, and $\Phi$ to denote the standard normal CDF. We write $X \ind Y$ to mean that $X$ and $Y$ are independent random variables.

For a joint distribution $P_{X, Y}$, we define the information density as 
\begin{align*}
    \imath_{P_{X, Y}}(x, y) &\coloneqq \log \frac{d P_{X,Y}}{d(P_X \times P_Y)}(x, y)\\
    &\,= \log \frac{d P_{Y|X = x}}{d P_Y}(y)
\end{align*}
for $(P_X \times P_Y)$-a.e. $(x, y)$, where we assume $P_{X, Y} \ll P_X \times P_Y$ and use the convention $\log(0) = -\infty$. Furthermore, for a given channel $P_{Y|X}$ and an arbitrary auxiliary output distribution $P_{\overline{Y}}$, we also define 
\begin{align*}
    \imath_{P_{Y|X}, P_{\overline{Y}}}(x, y) &\coloneqq \log \frac{d P_{Y|X = x}}{d P_{\overline{Y}}}(y) 
\end{align*}
for $P_{\overline{Y}}$-a.e. $y$ and for $x \in \mathcal{X}$ such that $P_{Y|X = x} \ll P_{\overline{Y}}$.

The fundamental limits of joint source--channel coding are governed by the coupling between the channel information density and the negative logarithm of the source distortion-ball probability. We call these the channel and source random variables, respectively. 

\subsection{Channel Random Variable}

Consider the channel input alphabet $\mathcal{X} = \{0,1 \}^m$ and the output alphabet $\mathcal{Y} = \mathcal{X} \cup \{ e\}$. The block erasure channel $P_{Y|X} = \operatorname{BLEC}(2^m, p)$ with block size $m$ and erasure probability $p$ is defined as 
\begin{align*}
    P_{Y|X}(y|x) &= \begin{cases}
        1 - p & y = x,\\
        p & y = e,\\
        0 & y \in \mathcal{X}, y \neq x.
    \end{cases}
\end{align*}
Let $P_{X^n} = \operatorname{Unif}( \mathcal{X}^n)$ and $P_{Y^n|X^n} = P^{(1)}_{Y|X} \times \cdots \times P^{(n)}_{Y|X}$, where $P_{Y|X}^{(i)} = \operatorname{BLEC}(2^m, p_i)$. The marginal distribution $P_{Y^n}$ induced by the joint $P_{X^n} \times P_{Y^n|X^n}$ is $P_{Y^n} = P_{Y_1} \times \cdots \times P_{Y_n}$, where 
\begin{align}
    P_{Y_i}(y_i) &= \begin{cases}
        p_i & y_i = e,\\
        \frac{1 -p_i}{2^m} & y_i \in \mathcal{X}. 
    \end{cases}
    \label{pyi}
\end{align}
Then the channel random variable is 
\begin{align}
    \imath_{P_{X^n, Y^n}}(X^n, Y^n) &= \sum_{i=1}^n \log \frac{P_{Y_i|X_i}(Y_i|X_i)}{P_{Y_i}(Y_i)}\\
    &= \log(2^m) \sum_{i=1}^n  \mathds{1}\{ Y_i \neq e \}, \label{i140b]]}
\end{align}
where the last equality holds $(P_{X^n} \times P_{Y^n|X^n})$-almost surely and $\mathds{1}$ is the indicator function. The sum in \eqref{i140b]]} has a Poisson binomial distribution with parameters $1-p_1, \ldots, 1-p_n$. Throughout the paper, we denote the PMF of this Poisson binomial distribution by $\{\pi_j\}_{j=0}^n$, where  
\begin{align}
   \pi_j &\coloneqq
\sum_{\substack{A \subseteq \{1,\dots,n\}\\ |A|=j}}
\prod_{i \in A} (1-p_i)
\prod_{k \notin A} p_k. \label{njerasureslej4}
\end{align}
In the special case of $p_1 = \cdots = p_n = p$, we define
\begin{align*}
    \widetilde{\pi}_j \coloneqq \binom{n}{j} (1-p)^j p^{n-j}.  
\end{align*}

\subsection{Source Random Variable}

Consider an i.i.d. Gaussian source $S^k \sim P_S^{\otimes k}$ where $P_S = \mathcal{N}(0, \sigma^2)$ and $0 < \sigma < \infty$. We denote the reproduction by $z^k \in \mathbb{R}^k$, and the reproduction distortion is measured by the mean squared error $ k^{-1} \|s^k - z^k\|^2$. For any distortion level \(d\in(0,\infty)\), the
rate-distortion function of the Gaussian source is
\begin{align}
    R(d) = \begin{cases}
         \frac12\log\frac{\sigma^2}{d} & \text{ if }0<d<\sigma^2,\\[0.5em]
        0 & \text{ if } d\ge \sigma^2.
    \end{cases} \label{rdf}
\end{align}
For $d \in (0, \sigma^2)$, we have $R(d) = I(P_S, P_{Z|S}^*)$, where  \cite[10.3.2]{Cover2006}
\begin{align}
    P_{Z|S}^*(\cdot|s)
    =
    \mathcal N\!\left(
        \left(1-\frac{d}{\sigma^2}\right)s,
        d\left(1-\frac{d}{\sigma^2}\right)
    \right),
\end{align}
and the reproduction marginal is $P_Z^*=\mathcal N(0,\sigma^2-d).$ Equivalently, the rate-distortion-achieving joint law $P_{S,Z}^* = P_S \times P_{Z|S}^*$ admits the backward
representation $S = Z + N$, where $Z\sim \mathcal N(0,\sigma^2-d), N\sim\mathcal N(0,d)$ and $Z \ind N$. Furthermore, for $d \in (0, \sigma^2)$, the Gaussian source rate-dispersion function \cite[Definition 7]{kostina_SC} in $\text{nats}^2$ is $V(d) = 1/2$. 

Define 
\begin{align*}
    B_d(s^k) &\coloneqq \left \{z^k \in \mathbb{R}^k : \|s^k - z^k\|^2 \leq k d \right\}.
\end{align*}
We refer to the negative log of the distortion-ball probability, $-\log P_{Z^k}(B_d(S^k)),$ as the \emph{source random variable}. For a source realization $S^k = s^k$, the \emph{distortion-ball probability} is the
probability that a single random reproduction codeword $Z^k \sim P_{Z^k}$ covers the fixed source
sequence $s^k$. When $Z^k$ is drawn uniformly from the surface of a sphere, $P_{Z^k}(B_d(s^k))$ is a spherical-cap probability. This is because geometrically, the distortion ball \(B_d(s^k)\) cuts out a spherical cap on
the reproduction sphere. To compactly write down the spherical-cap probability, we first define the \emph{normalized spherical-cap area function} as follows. For every $k \in \mathbb{Z}_{\geq 2}$,  define $L_k : \mathbb{R} \to [0, 1]$ as 
\begin{align*}
    L_k(t) \coloneqq \begin{cases}
0 & \text{ if } t\geq 1,\\[4pt]
1 & \text{ if } t\leq -1,\\[4pt]
\displaystyle \frac{1}{2}I_{1-t^2}\left( \frac{k-1}{2}, \frac{1}{2} \right) 
& \text{ if } 0 \leq t <1,\\
\displaystyle 1 - \frac{1}{2}I_{1-t^2}\left( \frac{k-1}{2}, \frac{1}{2} \right) & \text{ if } -1 < t < 0,
\end{cases}
\end{align*}
where $I_x(a, b)$ denotes the regularized incomplete beta function \cite[8.17.2]{NISTHandbook}. For $k = 1$, we extend the above definition as 
\begin{align*}
    L_1(t) &\coloneqq \begin{cases}
        0 & \text{ if } t > 1,\\
        1 & \text{ if } t \leq -1,\\
        \displaystyle \frac{1}{2} & \text{ if } -1 < t \leq 1. 
    \end{cases}
\end{align*}
We call it \emph{normalized} because $L_k(t)$ characterizes the spherical-cap area of a unit sphere at a given threshold $t$. Specifically, if \(U^k\) is uniformly distributed on the unit sphere in \(\mathbb R^k\), then \(\Pr[U_1\geq t]  =L_k(t)\); hence, \(L_k(t)\) is the fraction of the sphere occupied by the spherical cap with threshold \(t\). For each fixed $k \in \mathbb{Z}_{\geq 1}$, the function $L_k(t)$ is nonincreasing in $t$.

Lemma \ref{lemmaball} and the spherical-cap specialization of Proposition \ref{prop:normal-approximation} give useful results for the source random variable. Throughout the paper, $\Upsilon=\|S^k\|^2/(k\sigma^2)$ denotes the normalized source energy, and when $S^k=s^k$ we write its realized value as $\upsilon=\|s^k\|^2/(k\sigma^2)$. 
\begin{lemma}
Fix any $\sigma > 0$, $d \in (0, \sigma^2)$, $k \in \mathbb{Z}_{\geq 1}$ and $s^k\in\mathbb{R}^k \setminus \{ \mathbf{0}\}$. Let  \begin{align}
    P_{Z^k} = \operatorname{Unif}\left( \left \{z^k \in \mathbb{R}^k : \frac{1}{k}\sum_{i=1}^k z_i^2 = \sigma^2 - d \right \} \right). \label{pzkchoice}
\end{align} 
Then 
\begin{equation}
 \operatorname{Pr}\left(\|Z^k-s^k\|^2\leq kd\right) = \rho_{k, d}(\upsilon) \coloneqq L_k\left(\frac{\upsilon + 1 - 2 \delta}{2 \sqrt{\upsilon} \sqrt{1 - \delta}}\right), \label{rhokddef}
\end{equation}
where $\delta = \frac{d}{\sigma^2}$ and $\upsilon = \frac{\|s^k\|^2}{k \sigma^2}$. 
\label{lemmaball}
\end{lemma}
\textit{Proof:} The proof of Lemma \ref{lemmaball} is given in Appendix \ref{lemmaball_proof}. The result below is a straightforward corollary of Lemma \ref{lemmaball}. 
\begin{corollary}
\label{lemmaball_coro}
    Consider any sequence $s^k \neq \mathbf{0}$ and a random sequence $Z^k$ uniformly distributed on the surface of a sphere in $\mathbb{R}^k$ with radius $\sqrt{k r}$. Then for all $r > 0$ and $d > 0$,  
\begin{equation}
    \operatorname{Pr}\left(\|Z^k-s^k\|^2\leq kd\right) = L_k \left( \frac{\| s^k \|^2 + kr - kd}{2 \sqrt{kr}\, \|s^k \|} \right). \label{opt_int_L_K} 
\end{equation}
\end{corollary}

\textit{Geometric interpretation:} Lemma~\ref{lemmaball} gives an exact finite-dimensional expression for the
probability that a single reproduction codeword covers a fixed source
sequence. Since the reproduction distribution \(P_{Z^k}\) is rotationally
invariant, this probability depends on \(s^k\) only through its normalized
energy
\[
\upsilon=\frac{\|s^k\|^2}{k\sigma^2}.
\]
Geometrically, the distortion ball \(B_d(s^k)\) cuts out a spherical cap on
the reproduction sphere of radius \(\sqrt{k(\sigma^2-d)}\), and
\(\rho_{k,d}(\upsilon)\) is the ratio of the surface area of this cap to the total surface area of the reproduction sphere.

\textit{Random-coding analysis:} If a random codebook consists of \(M\) independent codewords
\(Z^k(1),\ldots,Z^k(M)\), each distributed according to \(P_{Z^k}\), then,
for every fixed \(s^k\) with normalized energy \(\upsilon\),
\[
\mathbb{P}\left[
\min_{1\leq m\leq M}\|Z^k(m)-s^k\|^2>kd
\,\middle|\, S^k=s^k
\right]
=
\bigl(1-\rho_{k,d}(\upsilon)\bigr)^M.
\]
Thus, the random-coding analysis is governed by the magnitude of
\(-\log \rho_{k,d}(\upsilon)\). Under the Gaussian source distribution, the
normalized source energy
\[
\Upsilon=\frac{\|S^k\|^2}{k\sigma^2} \stackrel{d}{=} \frac{\chi_k^2}{k},
\]
and thus $\Upsilon$ concentrates near \(1\). The next proposition gives a threshold-uniform normal approximation for functions of this energy that admit a local quadratic expansion. Its spherical-cap specialization (Corollary \ref{coro_specialize}) describes \(-\log \rho_{k,d}(\Upsilon)\), allowing also for a locally uniformly bounded perturbation.

\begin{proposition}[Normal approximation from a local quadratic expansion]
\label{prop:normal-approximation}
Let $\Upsilon\stackrel{d}{=}\chi_k^2/k$. Let $(f_k)_{k\geq1}$ be any sequence of measurable functions from $(0,\infty)$ to $\mathbb R\cup\{-\infty,+\infty\}$, and let $(b_k)_{k\geq1}$ be any sequence of deterministic real numbers. Suppose there exist positive constants $\eta_1\in(0,1)$, $k_1\geq1$, and $C_1<\infty$ such that, for every $k\geq k_1$ and $|\upsilon-1|\leq\eta_1$,
\begin{align}
\left|f_k(\upsilon)-b_k-\frac{k}{2}(\upsilon-1)\right|
&\leq C_1\left(1+k(\upsilon-1)^2\right).
\label{eq:local-quadratic-expansion}
\end{align}
Define the extended-real-valued random variables 
\begin{align}
    W_k = f_k(\Upsilon), \quad \quad \widehat W_k = \frac{W_k - b_k}{\sqrt{k/2}}. 
\end{align}
Then there exist constants $C>0$ and $k_0\geq k_1$, depending only on $\eta_1$, $k_1$, and $C_1$, such that, for $k\geq k_0$,
\begin{align}
\sup_{w\in\mathbb R}
\left|\mathbb P \left [  \widehat W_k \leq w \right ]-\Phi(w)\right|
&\leq\frac{C}{\sqrt{k}}.
\label{tronle}
\end{align}
\end{proposition}

\begin{corollary}[Spherical-cap specialization]
    Fix $\sigma>0$ and $d\in(0,\sigma^2)$, and let $S^k$ have i.i.d. $\mathcal N(0,\sigma^2)$ components. Recall the definition of $\rho_{k, d}(\upsilon)$ in \eqref{rhokddef}. Let each $\xi_k:(0,\infty)\to\mathbb R\cup\{-\infty,+\infty\}$ be measurable, and suppose that $-\log\rho_{k,d}(\upsilon)+\xi_k(\upsilon)$ is well-defined as an extended real number for every $k\geq1$ and $\upsilon>0$, with the convention $-\log0=+\infty$. Suppose that, for some $\eta_1\in(0,1)$, $k_1\geq1$, and $C_1<\infty$,
\begin{align}
\sup_{k\geq k_1}\sup_{|\upsilon-1|\leq\eta_1}|\xi_k(\upsilon)|
&\leq C_1.
\label{eq:spherical-perturbation-bound}
\end{align}
Set
\begin{align*}
\Upsilon&=\frac{\|S^k\|^2}{k\sigma^2},\\
W_k&=-\log\rho_{k,d}(\Upsilon)+\xi_k(\Upsilon),\\
b_k&=kR(d)+\frac12\log k.
\end{align*}
Then there exist constants $C$ and $k_0$, depending only on $\delta=d/\sigma^2$, $\eta_1$, $k_1$ and $C_1$, such that, for $k\geq k_0$,
\begin{align}
\sup_{w\in\mathbb R}
\left|\mathbb P\left[\frac{W_k-b_k}{\sqrt{k/2}}\leq w\right]-\Phi(w)\right|
&\leq\frac{C}{\sqrt{k}}.
\label{eq:spherical-cap-normal}
\end{align}
\label{coro_specialize}
\end{corollary}
\textit{Proof:} Proposition~\ref{prop:normal-approximation} and its spherical-cap specialization corollary are proved in Appendix~\ref{prop:normal-approximation_proof}.

\textit{Discussion of Proposition~\ref{prop:normal-approximation}:} The proposition gives a normal approximation for \emph{any} function of the normalized Gaussian source energy that satisfies the local expansion~\eqref{eq:local-quadratic-expansion}. The function need not be differentiable or monotone, and its behavior outside the specified neighborhood of $\upsilon=1$ is unrestricted. The permitted remainder has size \(O(1+k(\upsilon-1)^2)\), so the proposition applies to functions that share the same local linear behavior at \(\upsilon=1\) while differing at quadratic order. This flexibility will allow Proposition \ref{prop:normal-approximation} to be applied in both the achievability and converse analyses leading to Theorems~\ref{asymp_thm_achievability} and~\ref{thm_semiasympconverse}.

A useful sufficient condition for~\eqref{eq:local-quadratic-expansion} is the representation
\[
f_k(\upsilon)=b_k+\frac{k}{2}h(\upsilon)+r_k(\upsilon),
\]
where $h$ is twice continuously differentiable near $1$, $h(1)=0$, $h'(1)=1$, and $r_k$ is uniformly bounded on a fixed neighborhood of $1$. Taylor's theorem then gives~\eqref{eq:local-quadratic-expansion}; no differentiability of the remainder term $r_k(\upsilon)$ is needed.

Corollary \ref{coro_specialize} gives a \emph{robust} and \emph{threshold-uniform} normal approximation for the logarithm of the
spherical-cap probability $\rho_{k,d}(\Upsilon)$ defined in \eqref{rhokddef}. More precisely, it shows that 
\[
-\log \rho_{k,d}(\Upsilon)+\xi_k(\Upsilon),
\]
admits the deterministic centering
\[
b_k=kR(d)+\frac12\log k,
\]
fluctuations of order $\sqrt{k/2}$, and, after centering and scaling, a standard normal distribution up to a
Kolmogorov error of order $k^{-1/2}$. It is robust in the sense that the same approximation remains valid after adding any perturbation $\xi_k(\upsilon)$
that is uniformly bounded on a fixed neighborhood of $1$. On the other hand, threshold-uniform refers to the supremum in ~\eqref{tronle} and~\eqref{eq:spherical-cap-normal} that will
be essential in our application of this result: it permits the threshold to be replaced, in the subsequent
JSCC analysis, by the independent channel random variable while preserving the same $O(k^{-1/2})$ error.


Relevant source-coding comparisons for the spherical-cap specialization are
\cite[Lemma~2]{kostina_SC}, which gives a one-sided, high-probability refinement of the lossy AEP for the logarithm of a distortion-ball probability, and \cite[Theorem~40 and Appendix~K]{kostina_SC}, which derives a GMS rate expansion at a fixed target excess-distortion probability. Related spherical-cap geometry appears in the Gaussian source-coding converse of Palzer and Timo \cite{palzer_timo_converse}; see also \cite[Section~3.2.5]{palzer_thesis}. In contrast, Proposition \ref{prop:normal-approximation}
gives a normal approximation uniformly over all thresholds, which is the form needed for the subsequent
channel-uniform conditioning argument.


\textit{Proof Outline of Proposition~\ref{prop:normal-approximation}:} At a high level, the proof proceeds in four steps.
\begin{enumerate}
    \item Since $\Upsilon\stackrel{d}{=}\chi_k^2/k$, a Chernoff bound shows that $\Upsilon$ lies in a fixed neighborhood of $1$ with probability $1-O(e^{-\gamma k})$.

    \item On this neighborhood, the local expansion~\eqref{eq:local-quadratic-expansion} implies that 
    \[
    |\widehat W_k-U_k|\leq\frac{C(1+U_k^2)}{\sqrt{k}},
    \qquad
    U_k=\sqrt{\frac{k}{2}}(\Upsilon-1)
    \stackrel{d}{=} \frac{\chi_k^2-k}{\sqrt{2k}}.
    \]
    The Berry--Esseen theorem gives a standard normal approximation for $U_k$ with error $O(k^{-1/2})$.

    \item On a sufficiently small fixed neighborhood, the functions
    \[
    g_{k,\pm}(u)=u\pm\frac{C(1+u^2)}{\sqrt{k}}
    \]
    are strictly increasing over the corresponding range of $U_k$ and bound $\widehat W_k$ from below and above. We extend each function continuously, with slope $1$ outside this range, to an increasing bijection over $\mathbb R$. Their inverse functions, together with the bounded Gaussian density and its tail decay, transfer the normal approximation from $U_k$ to both bounding functions $g_{k, +}(U_k)$ and $g_{k, -}(U_k)$.
    \item Finally, a distribution-function sandwich transfers the approximation from the two bounding functions to $\widehat W_k$. The complement of the typical neighborhood is handled by the exponential Chernoff bound, yielding a uniform estimate over all $w\in\mathbb R$.
\end{enumerate}

\begin{table*}[h]
\begin{minipage}{\textwidth}
\begin{center}
\setcellgapes{2pt}
\makegapedcells
\renewcommand{\arraystretch}{1.12}

\begin{tabularx}{\textwidth}{|C{0.22\textwidth}|C{0.15\textwidth}|Y|}
\hline
\textbf{Category} & \textbf{Parameter} & \textbf{Description} \\
\hline

\multirow{3}{*}{\makecell{\textbf{Channel}\\\textbf{parameters}}}
& \(m\)
& BLEC block size; the channel input (resp. output) alphabet has size $2^m$ (resp. $2^m + 1$). \\
\cline{2-3}

& \(p\)
& Erasure probability for the memoryless block erasure channel
  \(\operatorname{BLEC}(2^m,p)\). \\
\cline{2-3}

& \(p_1,\ldots,p_n\)
& Erasure probabilities across the \(n\) channel uses of a nonstationary BLEC. \\
\hline

\makecell{\textbf{Source} \textbf{parameter}}
& \(\sigma^2\)
& Variance of the Gaussian source \(P_S=\mathcal{N}(0,\sigma^2)\). \\
\hline

\multirow{4}{*}{\makecell{\textbf{Design}\\\textbf{parameters}}}
& \(d\)
& Target distortion level. \\
\cline{2-3}

& \(\epsilon\)
& Target excess-distortion probability. \\
\cline{2-3}

& \(k\)
& Source blocklength. \\
\cline{2-3}

& \(n\)
& Channel blocklength. \\
\hline
\end{tabularx}

\end{center}
\caption{Joint source--channel coding parameters.}
\label{tab:parameters}
\end{minipage}
\end{table*}

An $(n, k)$ source--channel code is a pair of (possibly randomized) mappings $\operatorname{f} : \mathbb{R}^k \to \mathcal{X}^n$ and $\operatorname{g} : \mathcal{Y}^n \to \mathbb{R}^k$. An $(n, k)$ source--channel code is called $(d, \epsilon, n, k)$ if $\operatorname{Pr}(\|S^k - Z^k\|^2 > kd) \leq \epsilon$, where $S^k \sim P_S^{\otimes k}$, $Z^k = \operatorname{g}(Y^n)$, $X^n = \operatorname{f}(S^k)$ and $Y^n|X^n \sim P_{Y^n|X^n}$.

Achievability results give a sufficient condition in terms of $d, \epsilon, n$ and $k$ for an optimal source--channel code to be $(d, \epsilon, n, k)$. Converse results give a necessary condition for any code to be $(d, \epsilon, n, k)$. The excess-distortion event $\{\|S^k - Z^k\|^2 > k d \}$ is sometimes also called the error event. This is because the target excess-distortion probability $\epsilon$ plays a similar role in the fundamental bounds for JSCC as the average error probability does in channel coding. For example, compare \cite[(1)]{kostina_JSCC} for channel coding and \cite[(3)]{kostina_JSCC} for joint source--channel coding. 

\section{Achievability Results}

We first describe the random source--channel coding scheme underlying the achievability/existence result in Theorem \ref{thm_achievability}.    

\textit{Description of the source--channel code:} We denote the random $(n, k)$ source--channel code depending on the distortion level $d$ as $C_{n, k}^{(d)}$. Let
$\{ (\bar{X}_i^n, \bar{Z}_i^k)  \}_{i \in \mathbb{N}}$ be an infinite random codebook where each entry $(\bar{X}_i^n, \bar{Z}_i^k)$ is a pair of channel and reproduction sequences which are independently generated according to the product measure $ P_{X^n} \times P_{Z^k}$, where $P_{X^n} = \operatorname{Unif}( \mathcal{X}^n)$ and $P_{Z^k}$ is the spherical reproduction distribution in \eqref{pzkchoice}.
The codebook is independent of the source $S^k$ and is shared between the encoder and decoder of $C_{n, k}^{(d)}$. The encoder and decoder mappings are defined as follows: 
\begin{align}
   &s^k \mapsto \bar{X}_{i^\star}^n, \text{ where } i^\star = \begin{cases}
        \min \limits_{\bar{Z}_i^k \in B_d(s^k)} i & \text{ if } P_{Z^k}(B_d(s^k)) > 0,\\
        1 & \text{ otherwise,}
    \end{cases} \label{bec_enc}\\
    &y^n \mapsto \bar{Z}_{i^*}^k, \text{ where } i^* = \min_{ \bar{X}_{i, j} = y_j \forall j \text{ with } y_j \neq e} i. \label{bec_dec}
\end{align}

The results below are given for nonstationary block erasure channels. For the stationary case $p_1 = \cdots = p_n = p$, replace $\pi_j$ with $\widetilde{\pi}_j$ throughout. The results apply to both deterministic and random codes.  

\begin{theorem}[Non-asymptotic achievability]
Fix any $\sigma>0$, $d\in(0,\sigma^2)$ and
$m,n,k\in\mathbb Z_{\geq1}$. Let $\delta = d/\sigma^2$. Then there exists a $(d, \epsilon, n, k)$ source--channel code provided that   
\begin{align}
    \epsilon &\geq \sum_{j=0}^n \pi_j \int_0^\infty p_{\Upsilon}(\upsilon) \Psi_j(\upsilon; d)  d\upsilon, \label{achie2r794h}  
\end{align}
where 
\begin{align*}
    \Psi_j(\upsilon; d) &= \frac{1-\rho_{k, d}(\upsilon)}
{1+\bigl(2^{mj}-1\bigr)\rho_{k, d}(\upsilon)}
\end{align*}
and $p_{\Upsilon}(\upsilon)$ is a PDF given by 
\begin{align}
    p_{\Upsilon}(\upsilon) &= \frac{(k/2)^{k/2}}{\Gamma(k/2)} \upsilon^{k/2-1} e^{-k\upsilon/2}, \quad \upsilon \geq 0. \label{pzzz}
\end{align}
\label{thm_achievability}
\end{theorem}
\textit{Proof:} The proof of Theorem \ref{thm_achievability} is given in Appendix \ref{thm_achievability_proof}.

\textit{Discussion of Theorem \ref{thm_achievability}:} Theorem \ref{thm_achievability}'s non-asymptotic bound is stronger than the bounds obtained by simply particularizing the general non-asymptotic bounds \cite[Theorem 8]{kostina_JSCC} and \cite[Theorem 4]{alternative_oneshot_JSCC} to the GMS-BLEC case; this particularization entails choosing the channel input distribution $P_{X^n} = \operatorname{Unif}( \mathcal{X}^n)$ and the reproduction distribution $P_{Z^k}$ as in \eqref{pzkchoice}. First, note that \cite[Theorem 4]{alternative_oneshot_JSCC} is stronger than \cite[Theorem 8]{kostina_JSCC}. Second, particularizing the bound in \cite[Theorem 4]{alternative_oneshot_JSCC} to the GMS-BLEC case and then applying our Lemma \ref{lemmaball} gives us exactly the same result as Theorem \ref{thm_achievability} but with $\Psi_j(\upsilon; d)$ replaced with 
\begin{align}
    \overline{\Psi}_j(\upsilon; d) = \frac{1}{1 +  2^{m j} \rho_{k, d}(\upsilon)}. 
\end{align}
Clearly, $\overline{\Psi}_j(\upsilon; d) \geq \Psi_j(\upsilon; d)$ with equality if and only if $\rho_{k, d}(\upsilon) = 0$. This tightening can be relevant at short blocklengths; moreover, the RHS of \eqref{achie2r794h} is an exact expression for the ensemble-average excess-distortion probability of our achievability scheme rather than merely an upper bound. For large $k$, however, $\rho_{k,d}(\upsilon)$ is exponentially small for typical source realizations, so the improvement over $\overline{\Psi}_j$ is asymptotically negligible and is not the source of the third-order refinement discussed later.

The achievability scheme in $(\ref{bec_enc})$-$(\ref{bec_dec})$ can also be adapted for non-Gaussian sources so that a similar result to Theorem \ref{thm_achievability} can be proven for source--channel coding of arbitrary sources over block erasure channels; we state the adaptation without proof in Proposition \ref{general_sources_adapt}.
\begin{proposition}
    Fix any $d > 0$ and $m,n,k\in\mathbb Z_{\geq1}$. Let $Q_{S^k}$ and $Q_{Z^k}$ be arbitrary Borel probability measures on $\mathbb{R}^k$. Let the erasure probabilities $p_1,\ldots,p_n$ satisfy $p_\ell\in[0,1]$ for $\ell\in\{1,\ldots,n\}$. Consider the source--channel code in $(\ref{bec_enc})$-$(\ref{bec_dec})$ with $P_{Z^k}$ replaced by the arbitrary reproduction distribution $Q_{Z^k}$ and $B_d(s^k) = \{z^k: \operatorname{d}_k(s^k, z^k) \leq d \}$ for some arbitrary measurable distortion function $\operatorname{d}_k$. Let $S^k \sim Q_{S^k}$ be independent of the random codebook and the channel. Define $\rho_{k, d}(s^k) = Q_{Z^k}\left( B_d(s^k) \right)$. Then the ensemble-average excess-distortion probability of the code is equal to 
    \begin{align}
        \sum_{j = 0}^n \pi_j  \int Q_{S^k}(d s^k) \Psi_j(s^k; d), 
    \end{align}
    where 
\begin{align}
\Psi_j(s^k; d)
= \displaystyle
\frac{1-\rho_{k, d}(s^k)}
{1+\bigl(2^{mj}-1\bigr)\rho_{k, d}(s^k)}.
\end{align}
    \label{general_sources_adapt}
\end{proposition}

\begin{theorem}[Semi-asymptotic achievability] Fix $\sigma>0$, $d\in(0,\sigma^2)$, $\epsilon\in(0,1)$, and
$m\in\mathbb{Z}_{\geq 1}$. There exist constants
$k_0$ and $C$, depending only on
$\delta=d/\sigma^2$, $\epsilon$ and $m$, such that
for every $n\in\mathbb{Z}_{\geq 1}$, $p_1,\ldots,p_n\in[0,1]$ and $k\geq k_0$, there exists a $(d,\epsilon,n,k)$
source--channel code provided that 
\begin{align}
\log(2^m)\sum_{i=1}^n(1-p_i) - \frac{k}{2}\log\left(\frac{\sigma^2}{d}\right) &\geq \sqrt{\frac{k}{2} + \bigl(\log(2^m)\bigr)^2
\sum_{i=1}^n p_i(1-p_i) }
\,Q^{-1}(\epsilon) +\frac{1}{2}\log k + C.
\label{eq:gaussian-corollary-statement}
\end{align}
\label{asymp_thm_achievability}
\end{theorem}
\textit{Proof:} The proof of Theorem \ref{asymp_thm_achievability} is given in Appendix \ref{asymp_thm_achievability_proof}.

\textit{Discussion of Theorem \ref{asymp_thm_achievability}:} The starting point of the proof of Theorem \ref{asymp_thm_achievability} is Theorem \ref{thm_achievability}. The following two intermediate results from the proof of Theorem \ref{asymp_thm_achievability} are of independent interest. They also describe the proof outline of Theorem \ref{asymp_thm_achievability}. 

\begin{proposition}[cf. $(\ref{rer22})$]
Fix $\sigma>0$ and $d\in(0,\sigma^2)$. There exist constants
$k_0$ and $K$, depending only on
$\delta=d/\sigma^2$, such that for every
$m,n\in\mathbb{Z}_{\geq 1}$, every
$p_1,\ldots,p_n\in[0,1]$ and $k\geq k_0$, there exists a
$(d,\epsilon,n,k)$ source--channel code whenever
\begin{align}
\epsilon \geq \frac{K}{\sqrt{k}}+
\sum_{j=0}^n\pi_j Q\left(
\frac{j\log(2^m)-kR(d)-\frac12\log k}{\sqrt{k/2}}
\right).
\label{eq:semi-asymptotic-statement}
\end{align}
\label{semiasympt}
\end{proposition}

\begin{proposition}[cf. $(\ref{eq:Pe-normal})$]
    Fix $\sigma>0$, $d\in(0,\sigma^2)$ and
$m\in\mathbb{Z}_{\geq 1}$. Then there exist constants
$k_0$ and $D$, depending only on
$\delta=d/\sigma^2$ and $m$, such that
for every $n\in\mathbb{Z}_{\geq 1}$, $p_1,\ldots,p_n\in[0,1]$ and $k\geq k_0$, there exists a
$(d,\epsilon,n,k)$ source--channel code whenever
\begin{align}
\epsilon
&\geq \frac{D}{\sqrt{k}} + 
Q\left(
\frac{\mu_n-kR(d)-\frac12\log k}{S_{n,k}}
\right),
\label{eq:Pe2255-normal}
\end{align}
where 
\begin{align*}
    \mu_n &= \log(2^m)\sum_{i=1}^n(1-p_i),\\
    S_{n, k} &= \sqrt{
\frac{k}{2}
+
\bigl(\log(2^m)\bigr)^2
\sum_{i=1}^n p_i(1-p_i)
}. 
\end{align*}
\label{9y,n}
\end{proposition}

The proofs of Propositions~\ref{semiasympt} and~\ref{9y,n}
in fact upper bound the ensemble-average excess-distortion probability
of the random code $C_{n,k}^{(d)}$ in
\eqref{bec_enc}--\eqref{bec_dec} by the respective right-hand
sides of \eqref{eq:semi-asymptotic-statement} and
\eqref{eq:Pe2255-normal}. Hence, each of the inequalities $(\ref{eq:gaussian-corollary-statement})$, $(\ref{eq:semi-asymptotic-statement})$ and $(\ref{eq:Pe2255-normal})$ is a  sufficient condition for the source--channel code in \eqref{bec_enc}--\eqref{bec_dec} to be $(d, \epsilon, n, k)$ provided that $k$ is sufficiently large. Consequently, those are also sufficient conditions for the existence of a deterministic $(d, \epsilon, n, k)$ source--channel code.

The result in Proposition
\ref{semiasympt} is obtained by using the normal approximation
of the source random variable (Corollary \ref{coro_specialize}), while the distribution of the
channel random variable is retained exactly. On the other hand, Proposition \ref{9y,n} replaces the exact
Poisson-binomial mixture in Proposition \ref{semiasympt} by a
normal distribution approximation involving only the mean and variance of the
channel random variable. Importantly, this is not done by applying
a conventional central limit theorem to the channel random
variable in isolation. Such an approach would generally produce
an additional approximation error of
order $n^{-1/2}$ and the resulting approximation would therefore not be accurate uniformly over all $n \geq 1$. Instead, our proof first combines the channel random variable with the Gaussian fluctuation
arising from the normal approximation of the source random variable, whose variance is $k/2$, and then
applies Berry--Esseen to the combined random variable. In doing so, we first represent the normal-approximated source
term as a sum of arbitrarily many
independent Gaussian summands. This makes its contribution to the Berry–Esseen third-moment term arbitrarily small, leaving only the Bernoulli channel summands in $(\ref{i140b]]})$, whose total third absolute moment is upper bounded by a constant times their variance \(V_n\). The resulting Berry–Esseen error term is therefore proportional to
\[
    \frac{V_n}{(k/2+V_n)^{3/2}},
\]
which is uniformly $O(k^{-1/2})$ over all $V_n\geq 0$, where the channel dispersion $V_n$ is the variance of the channel random variable in $(\ref{i140b]]})$. Consequently, the constants $D$ and $k_0$ in Proposition \ref{9y,n} are
independent of the channel blocklength $n$ and of the
erasure profile $p_1,\ldots,p_n$. This Gaussian-smoothing
argument outlined above ensures that the additional normal
approximation in going from Proposition \ref{semiasympt} to Proposition \ref{9y,n} preserves the uniformity of the result w.r.t. $n$ and $p_1, \ldots, p_n$.

Finally, deriving the information balance condition $(\ref{eq:gaussian-corollary-statement})$ in Theorem \ref{asymp_thm_achievability} from Proposition \ref{9y,n} requires a uniform inversion at the target excess-distortion probability so as not to introduce dependence on $n$ and $p_1, \ldots, p_n$. 
 Let
\begin{align}
    x&=
    \frac{\mu_n-kR(d)-\frac{1}{2}\log k}
    {S_{n,k}}.
\end{align}
The error bound in Proposition \ref{9y,n} has the form
$Q(x)+D/\sqrt{k}$. An $O(k^{-1/2})$ error in probability
does not by itself imply an $O(1)$ correction in the information balance condition $(\ref{eq:gaussian-corollary-statement})$, since multiplying a
normalized correction by $S_{n,k}$ could potentially
introduce a dependence on $n$ and $p_1, \ldots, p_n$. Our proof handles this issue by separating the case
\(x\geq Q^{-1}(\epsilon)+1\), where the probability bound is already sufficiently below \(\epsilon\), from the case \(x<Q^{-1}(\epsilon)+1\), where one can show that \(S_{n,k}=O(\sqrt{k})\) uniformly in the channel parameters.

The familiar appearance of the source and
channel dispersion terms is not by itself the novelty of the
result in Theorem \ref{asymp_thm_achievability}. Rather, the novelty lies in establishing this
dispersion-type approximation uniformly over every positive
integer $n$ and every nonstationary erasure profile $p_1, \ldots, p_n$, together
with the refined $\frac{1}{2}\log k+O(1)$ third-order term in $(\ref{eq:gaussian-corollary-statement})$. In contrast, the third-order term obtained by directly specializing
the achievability part of \cite[Theorem~10]{kostina_JSCC} to the
GMS--BLEC pair is
\begin{align}
    \left(1+\frac{d^2}{\sigma^4-d^2}\right)\log k
    +\log\log k+O(1),                                      \label{m./}
\end{align}
for $d\in(0,\sigma^2)$. In writing \eqref{m./}, we used the fact
that $k=\Theta(n)$ in the coupled asymptotic regime of
\cite[Theorem~10]{kostina_JSCC}. One reason for our sharper
third-order term is that our asymptotic analysis is carried out after
specializing a one-shot result to the GMS--BLEC pair, and therefore exploits the exact structure of the GMS source and the erasure channel. By contrast, \cite[Theorem~10]{kostina_JSCC} applies
to a broad class of sources and channels.

The benefit of specializing the nonasymptotic bounds before carrying
out the asymptotic analysis was already observed in
\cite{kostina_JSCC}. In particular, the pair-specific analyses for
BMS--BSC and GMS--AWGN in
\cite[Theorem~15]{kostina_JSCC} and
\cite[Theorem~19]{kostina_JSCC}, respectively, give the improved
third-order term
\begin{align}
    \log k+\log\log k+O(1).                                \label{m./2}
\end{align}
This follows from \cite[(160)]{kostina_JSCC} and
\cite[(184)]{kostina_JSCC}, together with $k=\Theta(n)$ in the
coupled asymptotic regimes considered therein. These results,
however, apply only to stationary channel laws and are not uniform
in the channel blocklength and channel parameters. Furthermore, our GMS--BLEC
third-order term, $\frac12\log k+O(1)$, improves upon \eqref{m./2}. To explain the improvement, we first note that the GMS--AWGN analysis underlying \eqref{m./2}
can be decomposed as 
\begin{align}
\begin{split}
    \log k+\log\log k+O(1)
    &=
    \underbrace{\frac12\log k}_{\text{source}}
    +\underbrace{\frac12\log k}_{\text{threshold}}\\
    &\quad
    +\underbrace{\log\log k}_{\text{auxiliary }\gamma}
    +O(1).
    \end{split}
    \label{086bm./}
\end{align}
Our improved analysis features only the first of these logarithmic terms, which is already present in the source
random variable analysis. Indeed, the GMS
source-coding result \cite[Theorem 40]{kostina_SC}  obtains $\frac{1}{2}\log k + \log \log k + O(1)$. This is consistent with our own Corollary \ref{coro_specialize}, which features the $\frac{1}{2} \log k$ term.  Specifically, the $\frac{1}{2} \log k$ term arises directly from the
asymptotic behavior of the Gaussian spherical-cap probability $\rho_{k, d}(\Upsilon)$: its dominant exponential term\footnote{See \eqref{nmyknf}.} is multiplied by a factor
of order $k^{-1/2}$ for typical values of $\Upsilon$.

The extra $\log\log k$ penalty is attributable to the form of the
one-shot achievability bound used as the starting point for the
asymptotic analysis. The bound in
\cite[Theorem~8]{kostina_JSCC} contains both an auxiliary parameter
$\gamma$ and a separate error term $e^{1-\gamma}$. To make this
error term of order $k^{-1/2}$, the analysis takes $\gamma=1+\frac12\log k$ which adds an extra term $\log \gamma =\log\log k+O(1)$ in the
information-balance condition.
This $\log\log k$ penalty disappears when either the one-shot
achievability bound in our Theorem~\ref{thm_achievability} or that in
\cite[Theorem~4]{alternative_oneshot_JSCC} is used as the starting
point instead. On the other hand, the extra $\frac{1}{2} \log(k)$ penalty in \eqref{086bm./} comes from using a thresholding inequality on a typical event, e.g., an inequality of the form 
\begin{align*}
    \mathbb{E}\left [\frac{1}{1 + e^X} \right] \leq \operatorname{Pr}\left( e^{-X}\geq k^{-1/2}  \right) + \frac{1}{\sqrt{k}} 
\end{align*}
is used in \cite[Appendix D]{alternative_oneshot_JSCC}, while a similar thresholding step is used in \cite[(409)-(412)]{kostina_JSCC}. This $\frac{1}{2}\log k$ threshold penalty disappears by instead using the logistic representation \cite[Eq. (31)]{adeel_unknown_channels} 
\begin{align}
    \mathbb{E}\left [ \frac{1}{1 + e^{X}}\right] = \operatorname{Pr}\left(L \geq X \right), \label{498780343}
\end{align}
where $L \ind X$ is an independent standard logistic random variable. This technique yields improvement in a more general JSCC setting \cite{adeel_unknown_channels} and is not tied to the pair-specific analysis for GMS--BLEC.

The removal of the extra $\log\log k$ and $\frac12\log k$ terms
is a refinement of the achievability analysis and does not, by itself,
compare the optimal third-order coefficients of different
source--channel pairs. For the GMS--BLEC pair studied here,
Theorem~\ref{thm_semiasympconverse} additionally shows that the
remaining $\frac12\log k$ term is necessary. Consequently, our
achievability and converse conditions match through logarithmic third
order, with only an $O(1)$ difference in their information-balance
remainders. Analogous refinements for BMS--BSC or GMS--AWGN are not
ruled out by these comparisons.

\section{Converse Results}

In this section, we present converse results for the GMS-BLEC pair. 

\begin{theorem}[Non-asymptotic converse]
Fix any $\sigma>0$, $d\in(0,\sigma^2)$, $\epsilon\in(0,1)$ and $m,n,k\in\mathbb{Z}_{\geq 1}$. Let $\delta=d/\sigma^2$ and define
\begin{align}
\alpha_{k,d}(\upsilon)
&\coloneqq
\begin{cases}
1 & \text{ if }0\leq \upsilon\leq\delta,\\[3pt]
L_k\left(\sqrt{1-\frac{\delta}{\upsilon}}\right)
& \text{ if }\upsilon>\delta.
\end{cases}
\label{eq:maximal-source-cap}
\end{align}
If a $(d,\epsilon,n,k)$ source--channel code exists, then
\begin{align}
\epsilon \geq \sum_{j=0}^n\pi_j\Theta_j(k,d),
\label{m-}
\end{align}
where
\begin{align}
\Theta_j(k,d)
&\coloneqq\max\left\{
\int_0^\infty p_{\Upsilon}(\upsilon)\max\{0,1-2^{mj}\alpha_{k,d}(\upsilon)\}\,d\upsilon,\,
1-F_{\chi_k^2}\left(k\delta\,2^{2mj/k}\right)
\right\},
\label{eq:conditional-converse}
\end{align}
and $p_{\Upsilon}$ is the PDF in $(\ref{pzzz})$. The result applies to both deterministic and randomized codes.
\label{convo2}
\end{theorem}
\textit{Proof:} The proof of Theorem \ref{convo2} is given in Appendix \ref{convo2_proof}.

\textit{Discussion of Theorem \ref{convo2}:} We note that Theorem \ref{convo2} is at least as strong as the list-code converse \cite[Theorem 5]{kostina_JSCC} bound applied to the GMS--BLEC pair with the list measure taken to be the Lebesgue measure and the auxiliary output distribution taken to be $P_{\overline{Y}^n} = P_{\overline{Y}_1} \times \cdots \times P_{\overline{Y}_n}$, where 
\begin{align}
    P_{\overline{Y}_i}(y_i) &= \begin{cases}
        p_i & y_i = e,\\
        \frac{1 -p_i}{2^m} & y_i \in \mathcal{X}. 
    \end{cases}
    \label{2pyi}
\end{align}
We show this in Appendix \ref{list_code_converse} along with a recap of the list-code converse argument used to prove a converse for joint source--channel codes. 

Our converse proof first conditions on the erasure pattern. If exactly $j$ blocks are unerased, a deterministic decoder can produce at most $2^{mj}$ distinct reproduction sequences. Conditioned further on the normalized source energy, the Gaussian source is uniform on a sphere. Each reproduction sequence covers a spherical cap on this source sphere, and maximizing the cap area over the reproduction norm gives an upper bound on the fraction of source sequences that any one reproduction can cover. A union bound over the $2^{mj}$ possible reproductions, followed by averaging over the source energy and erasure pattern, gives a necessary condition for a $(d,\epsilon,n,k)$ source--channel code to exist. This is the first term in \eqref{eq:conditional-converse}. For the second term, we use the fact that successful reconstruction implies that $S^k$ lies in a distortion ball of one of the $2^{mj}$ possible reconstruction sequences. Maximizing this probability by using the total volume of the $2^{mj}$ reproduction balls obtains the second term. For a random decoder, the same argument holds by conditioning first on the random seed followed by averaging over the seeds.   

The spherical-cap component of Theorem~\ref{convo2}, i.e., the first term in \eqref{eq:conditional-converse}, alone is sufficient to obtain the \(\frac12\log k\) third-order term in Theorem \ref{thm_semiasympconverse}. For source energies in a fixed
neighborhood of $1$, the factorization
\eqref{eq:maximal-cap-factorization} gives
\begin{align*}
 -\log\alpha_{k,d}(\upsilon)
 &= kR(d)+\frac12\log k+\frac{k}{2}\log\upsilon+O(1),
\end{align*}
where the $O(1)$ remainder is uniform on that neighborhood. By comparison, the volume component is the tail probability
\begin{align*}
 1-F_{\chi_k^2}\left(k\delta\,2^{2mj/k}\right)
 &=\mathbb P\left[
 kR(d)+\frac{k}{2}\log\Upsilon>mj\log 2
 \right].
\end{align*}
Here, the source-side quantity in the volume bound is
$kR(d)+\frac{k}{2}\log\Upsilon$. Around the typical value
$\Upsilon=1$, it has the expansion
\begin{align*}
kR(d)+\frac{k}{2}(\Upsilon-1)
+O\!\left(k(\Upsilon-1)^2\right),
\end{align*}
and therefore has no additive $\frac12\log k$ term in its
centering. The volume
argument replaces the union of distortion balls by a centered ball
whose volume is $2^{mj}$ times that of one distortion ball and
thereby relaxes the angular covering constraint. In contrast, the spherical-cap quantity
$-\log\alpha_{k,d}(\Upsilon)$ is centered at
$kR(d)+\frac12\log k+O(1)$. 
The volume component of Theorem \ref{convo2} is nevertheless retained because it can strengthen the non-asymptotic result, and it also ensures that Theorem \ref{convo2} is at least as strong as the Lebesgue-measure list-code specialization discussed in Appendix \ref{list_code_converse}.

\begin{theorem}[Semi-asymptotic converse]
Fix $\sigma>0$, $d\in(0,\sigma^2)$, $\epsilon\in(0,1)$, and
$m\in\mathbb{Z}_{\geq 1}$. Then there exist constants
$k_0\in\mathbb{Z}_{\geq 1}$ and $C>0$, depending only on
$\delta = d/\sigma^2$, $\epsilon$ and $m$, such that for every $n\in\mathbb{Z}_{\geq 1}$,
$p_1,\ldots,p_n\in[0,1]$, and
$k\in\mathbb{Z}_{\geq 1}$ with $k\geq k_0$, if a
$(d,\epsilon,n,k)$ source--channel code exists, then
\begin{align}
\log(2^m)\sum_{i=1}^n(1-p_i)-kR(d)
&\geq
\sqrt{
\frac{k}{2}
+
\bigl(\log(2^m)\bigr)^2
\sum_{i=1}^n p_i(1-p_i)
}
\,Q^{-1}(\epsilon)
+\frac12\log k-C.
\label{eq:normal-approximation-converse}
\end{align}
\label{thm_semiasympconverse}
\end{theorem}
\textit{Proof:} The proof of Theorem \ref{thm_semiasympconverse} is given in Appendix \ref{thm_semiasympconverse_proof}.

\textit{Discussion of Theorem~\ref{thm_semiasympconverse}:}
Theorems~\ref{asymp_thm_achievability}
and~\ref{thm_semiasympconverse} determine the same logarithmic
third-order term, although the bounded remainder is not identified. We call Theorems~\ref{asymp_thm_achievability}
and~\ref{thm_semiasympconverse} semi-asymptotic because only the
source blocklength must exceed a threshold depending on the fixed
distortion ratio, target excess-distortion probability, and block size; the channel
blocklength and erasure profile remain unrestricted. This does not mean that the channel blocklength can remain fixed along a feasible sequence with \(k\to\infty\): for fixed \(m\), the necessary information-balance condition in Theorem \ref{thm_semiasympconverse}
forces \(n\) to grow as well. Rather, the significance of channel uniformity is that no separate large-\(n\) assumption is needed. Once \(k\) is sufficiently large, the bounded remainder terms in Theorems~\ref{asymp_thm_achievability} and~\ref{thm_semiasympconverse} are uniform over \(n\) and the erasure profile.

The matching $\frac{1}{2} \log k$ term has a common geometric origin. The achievability
analysis uses a fixed reproduction sphere, whereas the converse
maximizes the covered fraction of each source-energy sphere over all
reproduction norms. Nevertheless, the negative logarithms of the two spherical-cap probabilities have the same centering
$kR(d)+\frac12\log k$ and the same linear source-energy fluctuation
$\frac{k}{2}(\upsilon-1)$. Their differences beyond this common linear term are accommodated by the \(O(1+k(\upsilon-1)^2)\) remainder allowed in Proposition~\ref{prop:normal-approximation}.

Similar to the use of the logistic random-threshold identity \eqref{498780343} in achievability, the converse proof uses the exponential random-threshold representation $\max\{0,1-e^{-x}\}=\mathbb P[E<x]$, where $E$ is an independent
unit-mean exponential random variable. 
Source-induced Gaussian smoothing and the subsequent uniform
inversion then give~\eqref{eq:normal-approximation-converse} without
an additional logarithmic loss. The resulting
$\frac12\log k-O(1)$ converse correction is sharper than the
$O(1)$ corrections in the positive-distortion BMS--BSC and GMS--AWGN
converses of \cite[(159) and (183)]{kostina_JSCC}. However, this does not show that the BLEC intrinsically has a smaller optimal third-order
term than the BSC or the AWGN channel. While Theorems \ref{asymp_thm_achievability} and \ref{thm_semiasympconverse} establish the third-order optimality of the \(\frac12\log k\) term for the GMS--BLEC pair, the bounds established for the BMS--BSC and GMS--AWGN pairs in \cite{kostina_JSCC} do not identify their optimal third-order terms.

\section{Comparison with Separate source--channel coding}
\label{sec:sscc}

We first describe the source and channel codes separately, and then
analyze their concatenation through a shared random permutation.
For any $M\in\mathbb Z_{\geq1}$, we write $[M]=\{1,\ldots,M\}$.

\subsection{Source Code}

An $(M,k)$ source code is a pair of (possibly randomized) mappings
$\operatorname{f}_{\mathrm s}:\mathbb R^k\to[M]$ and
$\operatorname{g}_{\mathrm s}:[M]\to\mathbb R^k$.
An $(M,k)$ source code is called $(d,\epsilon_{\mathrm s},M,k)$ if
$\Pr(\|S^k-Z^k\|^2>kd)\leq\epsilon_{\mathrm s}$, where
$S^k\sim P_S^{\otimes k}$ and
$Z^k=\operatorname{g}_{\mathrm s}(\operatorname{f}_{\mathrm s}(S^k))$.

\textit{Description of the source code:}
Let $\mathbf B_{\mathrm s}=(\bar Z_1^k,\ldots,\bar Z_M^k)$ be a random
codebook whose entries are independently generated according to the
uniform spherical distribution $P_{Z^k}$ in~\eqref{pzkchoice}.
The codebook is independent of $S^k$ and is shared between the source
encoder and decoder. Their mappings are
\begin{align}
 s^k&\mapsto\operatorname{f}_{\mathrm s}(s^k)
 =\begin{cases}
 \min\{i\in[M]:\bar Z_i^k\in B_d(s^k)\}
     &\text{if this set is nonempty},\\
 1  &\text{otherwise},
 \end{cases}
 \label{eq:sscc-source-encoder}\\
 i&\mapsto\operatorname{g}_{\mathrm s}(i)=\bar Z_i^k.
 \label{eq:sscc-source-decoder}
\end{align}
Thus, an excess-distortion event occurs precisely when none of the
$M$ reproduction sequences covers the source realization.

\begin{proposition}[Non-asymptotic source-coding achievability]
\label{prop:sscc-source-achievability}
Fix any $\sigma>0$, $d\in(0,\sigma^2)$ and
$M,k\in\mathbb Z_{\geq1}$. The ensemble-average excess-distortion
probability of the source code
in~\eqref{eq:sscc-source-encoder}--\eqref{eq:sscc-source-decoder}
is equal to
\begin{align}
 A_{\mathrm s}(M;k,d)
 &\coloneqq\int_0^\infty p_{\Upsilon}(\upsilon)
       \bigl(1-\rho_{k,d}(\upsilon)\bigr)^M\,d\upsilon,
 \label{eq:sscc-source-achievability}
\end{align}
where $p_{\Upsilon}$ and $\rho_{k,d}$ are given
in~\eqref{pzzz} and~\eqref{rhokddef}, respectively. In particular, there exists a deterministic $(M,k)$ source code
with excess-distortion probability at most $A_{\mathrm s}(M;k,d)$.
Moreover, every fixed $z^k$ satisfying
$\|z^k\|^2=k(\sigma^2-d)$ has covering probability
\begin{align}
 \Pr(\|S^k-z^k\|^2\leq kd)
 &=b_{\mathrm s}(k,d)
 \coloneqq\int_0^\infty p_{\Upsilon}(\upsilon)
                    \rho_{k,d}(\upsilon)\,d\upsilon.
 \label{eq:sscc-one-ball}
\end{align}
\end{proposition}
\textit{Proof:} The proof is given in
Appendix~\ref{app:sscc-achievability}.

\textit{Discussion of Proposition~\ref{prop:sscc-source-achievability}:} The source code uses the same reproduction distribution as
Theorem~\ref{thm_achievability}, but with a finite codebook of size $M$.
This is also the spherical random-coding construction used in
\cite[Theorem~37]{kostina_SC}; there, the spherical-cap probability is
lower-bounded to obtain an achievability bound, whereas
Proposition~\ref{prop:sscc-source-achievability} uses its exact value
from Lemma~\ref{lemmaball} to obtain the exact ensemble-average
excess-distortion probability. Lastly, note that the single-reproduction covering probability \eqref{eq:sscc-one-ball} can equivalently be
written as
\begin{align}
 b_{\mathrm s}(k,d)
 &= F_{\text{noncentral}}(k \delta), \qquad\delta=d/\sigma^2,
 \label{eq:sscc-noncentral-ball}
\end{align}
where $F_{\text{noncentral}}$ is the CDF of the noncentral chi-squared distribution with $k$ degrees of freedom and noncentrality parameter $k(1-\delta)$. Indeed, $\|S^k-z^k\|^2/\sigma^2$ has this distribution whenever
$\|z^k\|^2=k(\sigma^2-d)$.

\subsection{Channel Code}

An $(n,M)$ channel code is a pair of (possibly randomized) mappings
$\operatorname{f}_{\mathrm c}:[M]\to\mathcal X^n$ and
$\operatorname{g}_{\mathrm c}:\mathcal Y^n\to[M]$.
An $(n,M)$ channel code is called $(\epsilon_{\mathrm c},n,M)$ if
$\Pr(\widehat U\ne U)\leq\epsilon_{\mathrm c}$, where
$U\sim\operatorname{Unif}([M])$,
$X^n=\operatorname{f}_{\mathrm c}(U)$,
$Y^n|X^n\sim P_{Y^n|X^n}$ and
$\widehat U=\operatorname{g}_{\mathrm c}(Y^n)$.
The message is independent of the coding randomness and the channel
erasures.

\textit{Description of the channel code:}
Let $\mathbf B_{\mathrm c}=(\bar X_1^n,\ldots,\bar X_M^n)$ be a random
codebook whose entries are independently generated according to
$P_{X^n}=\operatorname{Unif}(\mathcal X^n)$.
The codebook is independent of the message and channel erasures and
is shared between the channel encoder and decoder. Their mappings are
\begin{align}
 u&\mapsto\operatorname{f}_{\mathrm c}(u)=\bar X_u^n,
 \label{eq:sscc-channel-encoder}\\
 y^n&\mapsto\operatorname{g}_{\mathrm c}(y^n)
 =\min\{u\in[M]:\bar X_{u,\ell}=y_\ell
                    \text{ for every }\ell\text{ with }y_\ell\ne e\}.
 \label{eq:sscc-channel-decoder}
\end{align}
If the set in~\eqref{eq:sscc-channel-decoder} is empty, the decoder
outputs $1$. This case cannot occur when a codeword is transmitted
through the BLEC, since the transmitted codeword is always consistent
with the unerased blocks.

\begin{proposition}[Non-asymptotic channel-coding achievability]
\label{prop:sscc-channel-achievability}
Fix any $m,n,M\in\mathbb Z_{\geq1}$ and
$p_1,\ldots,p_n\in[0,1]$.
The ensemble-average error probability of the channel code
in~\eqref{eq:sscc-channel-encoder}--\eqref{eq:sscc-channel-decoder},
under a uniformly distributed message, is equal to
$1-H_{\mathrm A}(M;n)$, where
\begin{align}
 H_{\mathrm A}(M;n)
 &\coloneqq\sum_{j=0}^n\pi_j\,
       \frac{1-(1-2^{-mj})^M}{M2^{-mj}}.
 \label{eq:sscc-channel-achievability}
\end{align}
Here $\pi_j$ is given in~\eqref{njerasureslej4}, and
$1/M\leq H_{\mathrm A}(M;n)\leq1$.
In particular, there exists a deterministic $(n, M)$ channel code
with uniform-message average error probability at most
$1-H_{\mathrm A}(M;n)$.
\end{proposition}
\textit{Proof:} The proof is given in
Appendix~\ref{app:sscc-achievability}.

\subsection{Separate Source--Channel Code}

A direct concatenation of the component source and channel codes would pass the
source index $I=\operatorname{f}_{\mathrm s}(S^k)$ to the channel
encoder. However, $I$ need not be uniformly distributed, whereas
Proposition~\ref{prop:sscc-channel-achievability} evaluates the channel
code under a uniform message. We therefore insert an independent
uniform permutation $\Pi$ at the interface, which is shared between
the encoder and decoder. For every fixed source index
$i$, its permuted label $\Pi(i)$ is uniform on $[M]$. We call the resulting SSCC architecture \emph{symmetrized} SSCC. After averaging over the
permutation, every source index consequently experiences the uniform-message average correct-decoding probability of the channel code. 

\textit{Description of the symmetrized source--channel code:}
Fix the component source and channel codes as in the previous two subsections with the common index alphabet
$[M]$. Let $\Pi$ be a uniformly distributed
permutation of $[M]$, independent of the source, channel erasures and
component codes, and shared between the encoder and decoder.
The source--channel mappings are
\begin{align}
 s^k&\mapsto
 \operatorname{f}_{\mathrm c}
 \bigl(\Pi(\operatorname{f}_{\mathrm s}(s^k))\bigr),
 \label{eq:sscc-permuted-code}\\
 y^n&\mapsto
 \operatorname{g}_{\mathrm s}
 \bigl(\Pi^{-1}(\operatorname{g}_{\mathrm c}(y^n))\bigr).
 \label{eq:sscc-permuted-decoder}
\end{align}
The component source and channel codes are
fixed independently of $\Pi$, and a fresh permutation is used for each
transmission. For $M\geq2$, we generate the source
and channel codebooks described above independently of each other,
the source and the channel erasures.
For $M=1$, we replace the spherical reproduction by $\mathbf0$, so that
$Z^k=\mathbf0$ regardless of the channel output. The code is called $(d,\epsilon,n,k)$ if
$\Pr(\|S^k-Z^k\|^2>kd)\leq\epsilon$, with the probability also
averaged over $\Pi$. 


\begin{proposition}[Non-asymptotic achievability for symmetrized SSCC]
\label{prop:sscc-achievability}
Fix any $\sigma>0$, $d\in(0,\sigma^2)$,
$m,n,k,M\in\mathbb Z_{\geq1}$ and $p_1,\ldots,p_n\in[0,1]$.
For $M\geq2$, define
\begin{align}
 \eta_{\mathrm A}(M;n)
 &\coloneqq\frac{M H_{\mathrm A}(M;n)-1}{M-1},
 \label{eq:sscc-eta-achievability}\\
 A_{\mathrm{sep}}(M;n,k,d)
 &\coloneqq
 \eta_{\mathrm A}(M;n)A_{\mathrm s}(M;k,d)
 +\bigl(1-\eta_{\mathrm A}(M;n)\bigr)
       \bigl(1-b_{\mathrm s}(k,d)\bigr).
 \label{eq:sscc-achievability-bound}
\end{align}
For $M=1$, define
\begin{align}
 A_{\mathrm{sep}}(1;n,k,d)
 &\coloneqq1-F_{\chi_k^2}(k\delta),\qquad\delta=d/\sigma^2.
 \label{eq:sscc-one-message}
\end{align}
The ensemble-average excess-distortion probability of the symmetrized
SSCC construction given in \eqref{eq:sscc-permuted-code}-\eqref{eq:sscc-permuted-decoder}  is equal to 
\begin{align}
 \Pr(\|S^k-Z^k\|^2>kd)
 &=A_{\mathrm{sep}}(M;n,k,d),
 \label{eq:sscc-ensemble-exact}
\end{align}
where the probability averages over the source, channel erasures,
independent component codebooks and shared permutation.
Consequently, a symmetrized $(d,\epsilon,n,k)$ SSCC code with
deterministic component source and channel codes exists whenever
\begin{align}
 \epsilon\geq  \min_{M\in\mathbb Z_{\geq1}} A_{\mathrm{sep}}(M;n,k,d).
 \label{eq:sscc-achievability-condition}
\end{align}
The shared random permutation is retained in this existence statement.
\end{proposition}
\textit{Proof:} The proof is given in
Appendix~\ref{app:sscc-achievability}.

\textit{Discussion of Proposition~\ref{prop:sscc-achievability}:} In the exact excess-distortion probability calculation
in~\eqref{eq:sscc-ensemble-exact}, a channel decoding error is not counted as
an excess-distortion event if the final reproduction meets the
distortion constraint. Counting only transmissions for which the source is covered and the
selected index is decoded correctly instead gives the simpler, but
weaker, sufficient condition
\begin{align}
 \epsilon\geq \min_{M\in\mathbb Z_{\geq 2}} \left \{
 1-\bigl(1-A_{\mathrm s}(M;k,d)\bigr)H_{\mathrm A}(M;n) \right\}.
 \label{eq:sscc-product-achievability}
\end{align}
This bound uses the same source/channel concatenation structure as the conventional
SSCC achievability result in~\cite[Theorem~6]{kostina_JSCC}, but it differs in two ways. A first difference is in the channel-coding criterion:
\cite[Theorem~6]{kostina_JSCC} uses a worst-case average channel
error probability, where the worst case is over the message distribution, whereas \eqref{eq:sscc-product-achievability} evaluates the average channel
error probability (specifically the average correct-decoding probability $H_{\mathrm A}(M;n)$) under a uniform message distribution.
The shared permutation is what permits the uniform-message
channel success probability $H_{\mathrm A}(M;n)$ to be used here:
for every fixed source index, its permuted channel label is uniform
on $[M]$, even though the source index itself need not be uniformly
distributed. A second difference is how the source- and channel-coding
performances are combined. The bound in
\cite[Theorem~6]{kostina_JSCC} has an additive union-bound form, obtained by
adding the source excess-distortion probability and the channel
decoding error probability. In contrast,
\eqref{eq:sscc-product-achievability} retains the joint success event
that the source is covered and its selected index is decoded
correctly, leading to the product
$\bigl(1-A_{\mathrm s}(M;k,d)\bigr)H_{\mathrm A}(M;n)$.

\section{Numerical Evaluation}
\label{sec:numerical-evaluation}

Define $A_n(k)$ to be the RHS of~\eqref{achie2r794h} and $C_n(k)$
to be the RHS of~\eqref{m-}, with their dependence on the other
parameters in Table~\ref{tab:parameters} left implicit. Define
\begin{align}
 k_A(n)
 &=\sup\{k\in\mathbb Z_{\geq1}:A_n(k)\leq\epsilon\},
 \label{kapossibleconfusion}\\
 k_C(n)
 &=\sup\{k\in\mathbb Z_{\geq1}:C_n(k)\leq\epsilon\},\nonumber
\end{align}
with the convention that $\sup\varnothing=0$ throughout this section.
Thus, $k_A(n)$ is the largest source blocklength certified achievable by the achievability bound in Theorem \ref{thm_achievability}, while $k_C(n)$ is the largest source blocklength not ruled out by the converse result in Theorem \ref{convo2}.
Let
\begin{align*}
 k^\star(n)
 &=\sup\{k\in\mathbb Z_{\geq1}:\text{a $(d,\epsilon,n,k)$
 source--channel code exists}\}
\end{align*}
denote the optimum. Then $k_A(n)\leq k^\star(n)\leq k_C(n)$.
For a fixed erasure-probability sequence $(p_i)_{i\geq1}$,
$C_n(k)$ is nonincreasing in $n$, so $k_C(n)$ is nondecreasing
in $n$, although $k_C(n)/n$ need not be nondecreasing.

Define also the JSCC finite-blocklength normal approximation based on Theorems \ref{asymp_thm_achievability} and \ref{thm_semiasympconverse} as   
\begin{align*}
    \widetilde{k}(n) &= \sup \left \{k \in \mathbb{Z}_{\geq 1}: \log(2^m)\sum_{i=1}^n(1-p_i)
-
\frac{k}{2}\log\left(\frac{\sigma^2}{d}\right)
\geq \right . \\
& \quad \quad \quad \quad  \left . 
\sqrt{
\frac{k}{2}
+
\bigl(\log(2^m)\bigr)^2
\sum_{i=1}^n p_i(1-p_i)
}
\,Q^{-1}(\epsilon)
+
\frac{1}{2}\log k  \right \}.
\end{align*}
This is the common third-order approximation furnished by Theorems~\ref{asymp_thm_achievability} and~\ref{thm_semiasympconverse} after omitting their bounded remainder constants. Because those constants are omitted, it is neither a
certified achievability bound nor a certified converse bound and need
not lie between $k_A(n)$ and $k_C(n)$ at finite blocklength.

For the SSCC construction in Proposition~\ref{prop:sscc-achievability},
define
\begin{align}
 A_n^{\mathrm{sep}}(k)
 &\coloneqq\min_{M\in\mathbb Z_{\geq1}}
                  A_{\mathrm{sep}}(M;n,k,d),
 \label{eq:sscc-numerical-error}\\
 k_A^{\mathrm{sep}}(n)
 &\coloneqq\sup\{k\in\mathbb Z_{\geq1}:
                  A_n^{\mathrm{sep}}(k)\leq\epsilon\}.
 \label{eq:sscc-numerical-blocklength}
\end{align}
The minimum over $M$ is attained, as shown in
Appendix~\ref{app:sscc-achievability}.
The quantity $k_A^{\mathrm{sep}}(n)$ is the largest source blocklength
certified achievable by this construction after optimizing the
interface size $M$, including the $M=1$ choice in~\eqref{eq:sscc-one-message}.
In particular, $k_A^{\mathrm{sep}}(n)\leq k^\star(n)$.

The corresponding rates, measured in source samples per channel
block, are
\begin{align}
 R_A(n)&=\frac{k_A(n)}{n},
 \quad\quad\quad R_C(n)=\frac{k_C(n)}{n},\nonumber\\
 \widetilde R(n)&=\frac{\widetilde k(n)}{n}, \quad\quad\quad\,
 R_A^{\mathrm{sep}}(n)=\frac{k_A^{\mathrm{sep}}(n)}{n}.
 \label{eq:sscc-numerical-rates}
\end{align}
Figure~\ref{RA_vs_RC} plots these four rates for $m=100$,
$\sigma=1$, $d=0.1$, $\epsilon=0.01$, and a stationary BLEC with
$p_i=p=0.1$, for $n=1,\ldots,32$.
The JSCC achievability and converse curves are relatively close
except at the shortest channel blocklengths. The normal approximation
lies above the converse curve over the plotted range, so omitting
the bounded remainder terms gives an optimistic approximation for
these parameters.
The JSCC achievable rate exceeds the SSCC achievable benchmark
for $n\geq3$ in the figure. Both achievable curves are obtained
from exact ensemble-average excess-distortion probabilities of their respective random coding schemes. Note, however, that this is a comparison
of the specified random-coding constructions; the SSCC curve does
not characterize the optimal performance of separate source--channel
coding.

For a memoryless, stationary BLEC, the asymptotically optimal rate is
\begin{align*}
 \lim_{n\to\infty}\frac{k^\star(n)}{n}
 &=\frac{C}{R(d)}
 =\frac{(1-p)\log(2^m)}{\frac12\log(\sigma^2/d)},
\end{align*}
which is approximately $54.1854$ source samples per channel block
for the plotted parameters.

\begin{figure}[H]
 \centering
 \includegraphics[width=12cm]{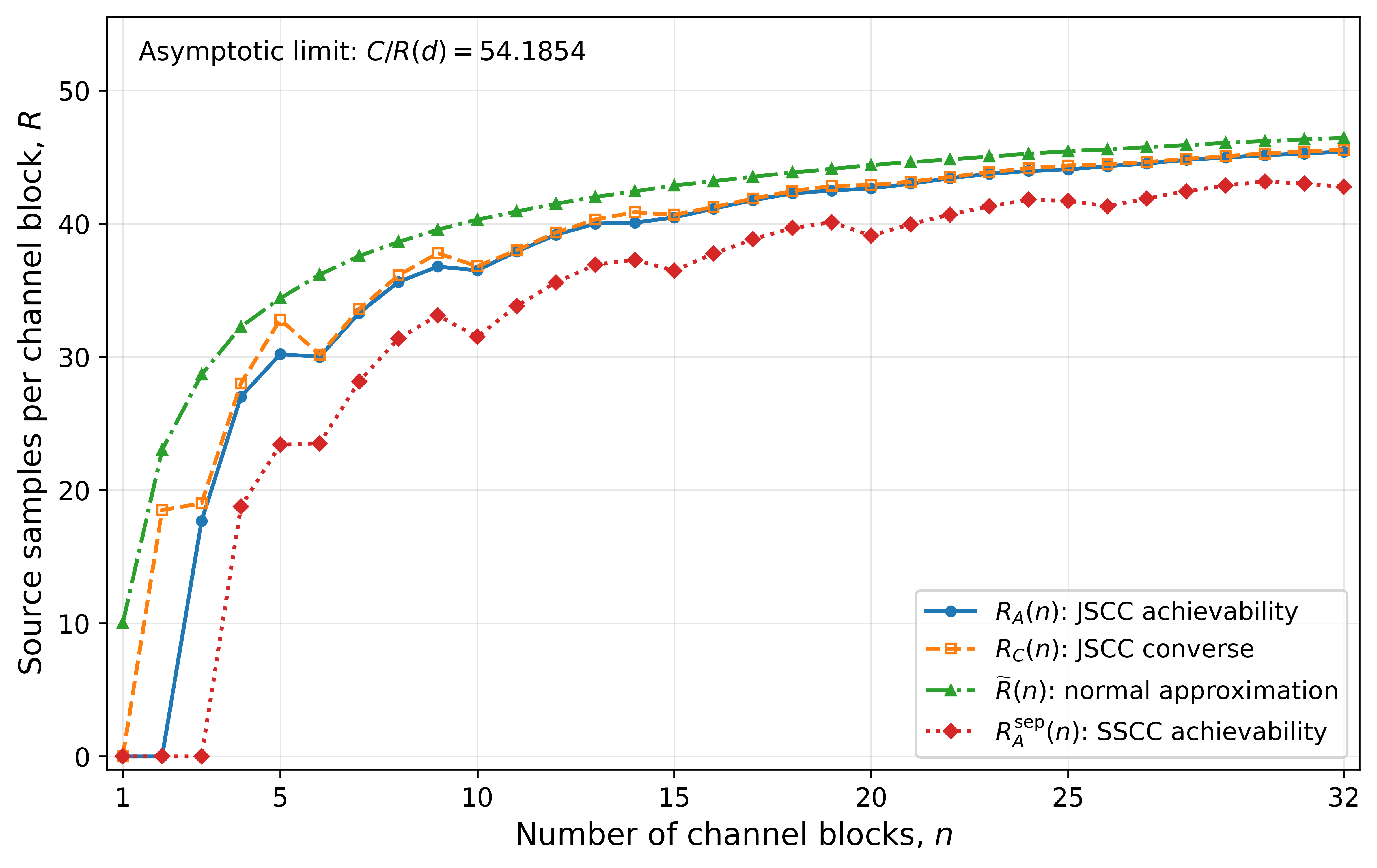}
 \caption{Nonasymptotic JSCC achievable rate $R_A(n)$ and converse
 rate $R_C(n)$ from Theorems~\ref{thm_achievability}
 and~\ref{convo2}, their common third-order normal approximation
 $\widetilde R(n)$, and the symmetrized SSCC achievable rate
 $R_A^{\mathrm{sep}}(n)$ from
 Proposition~\ref{prop:sscc-achievability}, optimized over the
 interface size $M$. The source is Gaussian and the BLEC is stationary,
 with $m=100$, $\sigma=1$, $d=0.1$, $\epsilon=0.01$, and $p=0.1$.}
 \label{RA_vs_RC}
\end{figure}

\appendices 

\section{Proof of Lemma \ref{lemmaball} \label{lemmaball_proof}}

Let
$r=\sqrt{k(\sigma^2-d)}$. Then $Z^k$
is uniform on the sphere of radius $r$. Let \(\rho=\|s^k\|\). Since the distribution of \(Z^k\) is invariant under orthogonal transformations and Euclidean distance is rotation invariant, we may apply an orthogonal transformation sending \(s^k\) to \((\rho,0,\ldots,0)\). Let $Z^k=rU^k$, where $U^k$ is uniform on the
unit sphere $S^{k-1}$. Then
\begin{align}
    \|Z^k-s^k\|^2
    &=
    r^2+\rho^2-2r\rho U_1.
\end{align}
Therefore, since $r>0$ and $\rho>0$,
\begin{align}
    \|Z^k-s^k\|^2\leq kd
    &\iff
    r^2+\rho^2-2r\rho U_1 \leq kd \\
    &\iff
    U_1 \geq t,
\end{align}
where
\begin{align}
    t
    &=
    \frac{r^2+\rho^2-kd}{2r\rho}\\
    &= \frac{\upsilon+1-2\delta}
    {2\sqrt{\upsilon}\sqrt{1-\delta}},
\end{align}
where $\upsilon$ and $\delta$ are defined in the statement of Lemma 1. Hence, $\operatorname{Pr}\left(\|Z^k-s^k\|^2\leq kd\right) = \operatorname{Pr}(U_1\geq t)$.

For the case $k=1$, we have $S^0=\{-1,1\}.$ Consequently, $U_1$ is uniformly distributed on $\{-1,1\}$. It follows that
\begin{align}
    \operatorname{Pr}(U_1\geq t)
    &=
    \begin{cases}
        0 & \text{ if } t > 1,\\
        1 & \text{ if } t \leq -1,\\
        \frac{1}{2} & \text{ if } -1 < t \leq 1. 
    \end{cases}
\end{align}
For $k \in \mathbb{Z}_{\geq 2}$,
the first coordinate $U_1$ has density
\begin{align}
    f_{U_1}(u)
    =
    \frac{\Gamma(k/2)}
    {\sqrt{\pi}\Gamma((k-1)/2)}
    (1-u^2)^{(k-3)/2},
    \qquad -1\leq u\leq 1.
\end{align}
Therefore,
\begin{align}
    \operatorname{Pr}(U_1\geq t) = 
    \begin{cases}
    0 & \text{if } t\geq 1,\\[4pt]
    1 & \text{if } t\leq -1,\\[4pt]
    \displaystyle
    \frac{\Gamma(k/2)}
    {\sqrt{\pi}\Gamma((k-1)/2)}
    \int_{t}^1 (1-u^2)^{(k-3)/2}\,du,
    & \text{if } -1<t<1.
    \end{cases}
\end{align}
It remains to rewrite the integral in terms of the regularized incomplete beta
function. For $0\leq t<1$,
\begin{align}
    \frac{1}{2}
    I_{1-t^2}
    \left(\frac{k-1}{2},\frac{1}{2}\right)
    &=
    \frac{\Gamma(k/2)}{2\Gamma((k-1)/2)\Gamma(1/2)}
    \int_0^{1-t^2}
    s^{\frac{k-3}{2}}(1-s)^{-1/2}\,ds.
\end{align}
Using the change of variables $u=\sqrt{1-s}$, we obtain, for $0 \leq t < 1$,
\begin{align}
    \frac{1}{2}
    I_{1-t^2}
    \left(\frac{k-1}{2},\frac{1}{2}\right)
    &=
    \frac{\Gamma(k/2)}
    {\sqrt{\pi}\Gamma((k-1)/2)}
    \int_{t}^1 (1-u^2)^{(k-3)/2}\,du. \label{mnlyt}
\end{align}

For $-1<t<0$, by symmetry of the density of $U_1$,
\begin{align}
    \operatorname{Pr}(U_1\geq t)
    &=
    1-\operatorname{Pr}(U_1<t) \\
    &=
    1-
    \frac{\Gamma(k/2)}
    {\sqrt{\pi}\Gamma((k-1)/2)}
    \int_{-1}^{t} (1-u^2)^{(k-3)/2}\,du \\
    &=
    1-
    \frac{\Gamma(k/2)}
    {\sqrt{\pi}\Gamma((k-1)/2)}
    \int_{-t}^{1} (1-u^2)^{(k-3)/2}\,du.
\end{align}
Since $-t\in(0,1)$, the previous identity $(\ref{mnlyt})$ gives
\begin{align}
    \operatorname{Pr}(U_1\geq t)
    &=
    1-
    \frac{1}{2}
    I_{1-t^2}
    \left(\frac{k-1}{2},\frac{1}{2}\right).
\end{align}
Combining the above cases and recalling the definition of $L_k$, we obtain
\begin{align}
    \operatorname{Pr}\left(\|Z^k-s^k\|^2\leq kd\right)
    &=
    L_k\left(
    \frac{\upsilon+1-2\delta}
    {2\sqrt{\upsilon}\sqrt{1-\delta}}
    \right)\\
    &=
    \rho_{k,d}(\upsilon).
\end{align}

\section{Proof of Proposition \ref{prop:normal-approximation} \label{prop:normal-approximation_proof}}

Define $s_k = \sqrt{k/2}$ and $U_k = s_k(\Upsilon-1)$. Since $\Upsilon\stackrel{d}{=}\chi_k^2/k$,
\begin{align*}
U_k\stackrel{d}{=} \frac{1}{\sqrt{k}}\sum_{i=1}^k
\frac{G_i^2-1}{\sqrt{2}},
\end{align*}
where the $G_i$'s are i.i.d. standard normal random variables. Hence, the Berry--Esseen theorem gives
\begin{align}
\sup_{u\in\mathbb{R}}
\left|
\mathbb{P}[U_k\leq u]-\Phi(u)
\right|
\leq\frac{C}{\sqrt{k}}.
\label{eq:berrmuy-esseen}
\end{align}

Let $a=\max\{1,2C_1\}$ and choose
\begin{align*}
    \eta=\min\left\{\eta_1,\frac12,\frac{1}{8a}\right\}.
\end{align*}
Let $\mathcal I=[1-\eta,1+\eta]$ and
$\mathcal G=\{\Upsilon\in\mathcal I\}$. By the usual Chernoff bound for $\chi_k^2$,
\begin{align}
    \mathbb P(\mathcal G^c)\leq Ce^{-\gamma k}
    \label{cher86noffG}
\end{align}
for constants $C$ and $\gamma>0$.
\begin{remark}
    In $(\ref{cher86noffG})$ and throughout the proof of the proposition, any unspecified constants such as $C$, $\gamma$ and $k_0$ depend only on $\eta_1$, $k_1$ and $C_1$. The values of these constants may change from one occurrence to the next.
\end{remark}

Since $k(\Upsilon-1)^2=2U_k^2$, the hypothesis $(\ref{eq:local-quadratic-expansion})$ gives, on the event $\mathcal G$ and for $k\geq k_1$,
\begin{align}
    |\widehat W_k-U_k|
    \leq \frac{C_1(1+2U_k^2)}{s_k}
    \leq \frac{a(1+U_k^2)}{s_k}
\end{align}
and hence,  
\begin{align}
    U_k -\frac{a(1+U_k^2)}{s_k} \leq \widehat W_k \leq U_k + \frac{a(1+U_k^2)}{s_k}. \label{eq:W-T-caeomparison}
\end{align}
Note that
\begin{align}
    \mathcal G=\{|\Upsilon-1|\leq\eta\}
    =\{|U_k|\leq\eta s_k\}.
    \label{1-a0uf}
\end{align}
In view of \eqref{eq:W-T-caeomparison}, we define the two "envelope" functions
\begin{align}
    g_{k,\pm}(u)=u\pm\frac{a(1+u^2)}{s_k},
    \qquad u\in\mathcal D,
    \label{eq:quadratic-envelopes}
\end{align}
where $\mathcal D=[-\eta s_k,\eta s_k]$. The choice of $\eta$ ensures that the  derivatives of both envelope functions on $\mathcal D$ lie between $3/4$ and $5/4$. Extend each function continuously to $\mathbb R$ by an affine function of slope $1$ on each of the intervals $(-\infty,-\eta s_k)$ and $(\eta s_k,\infty)$. The extended functions are strictly increasing and satisfy
\begin{align}
    \frac34
    \leq
    \frac{g_{k,\pm}(v)-g_{k,\pm}(u)}{v-u}
    \leq\frac54,
    \qquad u<v.
    \label{eq:envelope-slope-bounds}
\end{align}
Moreover, the extensions satisfy
\begin{align}
    |g_{k,\pm}(u)-u|
    \leq\frac{a(1+u^2)}{s_k},
    \qquad u\in\mathbb R,
    \label{eq:g-iden867tity-comparison}
\end{align}
since outside $\mathcal D$ the difference $g_{k,\pm}(u)-u$ is constant on each of the two intervals and has absolute value $a(1+\eta^2s_k^2)/s_k$.

We next transfer the normal approximation for \(U_k\) in \eqref{eq:berrmuy-esseen} to the two auxiliary random variables \(g_{k,+}(U_k)\) and \(g_{k,-}(U_k)\). We show that each satisfies the same \(O(k^{-1/2})\) Kolmogorov bound as in \eqref{eq:berrmuy-esseen}. Using \eqref{eq:W-T-caeomparison}, these estimates will later be used to bracket the distribution function of \(\widehat W_k\) itself. 

Fix either sign and write $g_k$ for the corresponding function. Since $g_k$ is strictly increasing and its range is $\mathbb R$, for every $w\in\mathbb R$ there is a unique $u_w\in\mathbb R$ such that $g_k(u_w)=w$. Hence,
\begin{align}
    \mathbb P[g_k(U_k)\leq w]=\mathbb P[U_k\leq u_w].
    \label{eq:ehevent-comparison}
\end{align}

We next compare $\Phi(w)$ and $\Phi(u_w)$. From $(\ref{eq:g-iden867tity-comparison})$,
\begin{align}
    |w-u_w|\leq\frac{a(1+u_w^2)}{s_k}.
    \label{eq:x3-y-bound}
\end{align}
Also, $|g_k(0)|=a/s_k$. \textit{Case 1:} If $|u_w|\leq4a/s_k$, then $(\ref{eq:x3-y-bound})$ and the bound $\sup_u\Phi'(u)\leq1/\sqrt{2\pi}$ give
\begin{align*}
    |\Phi(w)-\Phi(u_w)|\leq\frac{C}{s_k}.
\end{align*}
\textit{Case 2:} If $|u_w|>4a/s_k$, then $(\ref{eq:envelope-slope-bounds})$ shows that $w$ and $u_w$ have the same sign and
\begin{align*}
    |w|\geq\frac34|u_w|-\frac{a}{s_k}
    \geq\frac{|u_w|}{2}.
\end{align*}
Therefore in \textit{Case 2}, every point $u$ between $w$ and $u_w$ satisfies $|u|\geq|u_w|/2$. The mean value theorem gives
\begin{align*}
|\Phi(w)-\Phi(u_w)|
&\leq
|w-u_w|
\sup_{u\text{ between }w\text{ and }u_w}
\frac{1}{\sqrt{2\pi}}e^{-u^2/2}\\
&\leq
\frac{C}{s_k}(1+u_w^2)e^{-u_w^2/8}\\
&\leq\frac{C}{s_k}.
\end{align*}
Thus, in all cases,
\begin{align}
    \sup_{w\in\mathbb R}|\Phi(w)-\Phi(u_w)|
    \leq\frac{C}{\sqrt{k}}.
    \label{eq:normal-2cdf-comparison}
\end{align}
Combining $(\ref{eq:berrmuy-esseen})$, $(\ref{eq:ehevent-comparison})$ and $(\ref{eq:normal-2cdf-comparison})$, we obtain, for either sign,
\begin{align}
    \sup_{w\in\mathbb R}
    \left|\mathbb P[g_{k,\pm}(U_k)\leq w]-\Phi(w)\right|
    \leq\frac{C}{\sqrt{k}}.
    \label{eq:envelope-normal}
\end{align}

On the event $\mathcal G$, $(\ref{eq:W-T-caeomparison})$ and the definitions \eqref{eq:quadratic-envelopes} imply that 
$g_{k,-}(U_k)\leq\widehat W_k\leq g_{k,+}(U_k)$. Consequently, for every $w\in\mathbb R$,
\begin{align*}
    \mathbb P[g_{k,+}(U_k)\leq w]-\mathbb P(\mathcal G^c)
    &\leq\mathbb P[\widehat W_k\leq w]\\
    &\leq\mathbb P[g_{k,-}(U_k)\leq w]
    +\mathbb P(\mathcal G^c).
\end{align*}
Combining these inequalities with $(\ref{cher86noffG})$ and $(\ref{eq:envelope-normal})$ proves $(\ref{tronle})$.

\subsection*{Spherical-cap specialization (Corollary \ref{coro_specialize})}

We now prove $(\ref{eq:spherical-cap-normal})$ under the local boundedness condition $(\ref{eq:spherical-perturbation-bound})$. Let $\delta=d/\sigma^2$ and define
\begin{align*}
    t(\upsilon)&=
    \frac{\upsilon+1-2\delta}{2\sqrt{\upsilon}\sqrt{1-\delta}}.
\end{align*}
Since $t(1)=\sqrt{1-\delta}\in(0,1)$, we can choose $0<\eta\leq\min\{\eta_1,1/2\}$ such that whenever $|\upsilon-1|\leq\eta$, the value of $t(\upsilon)$ lies in a compact subinterval of $(0,1)$. In this specialization, let $\mathcal I=[1-\eta,1+\eta]$ and $\mathcal G=\{\Upsilon\in\mathcal I\}$. Here $\eta$ depends only on $\delta$ and $\eta_1$. Unspecified constants such as $C$, $\gamma$ and $k_0$ may now depend on $\delta$, $\eta_1$, $k_1$ and $C_1$, and may change from one occurrence to the next.

For $k\geq2$ and $0<t<1$, we rewrite $L_k(t)$ as
\begin{align}
    L_k(t)
    &=
    \frac{1}{2B\left(\frac{k-1}{2},\frac12\right)}
    \int_0^{1-t^2}u^{(k-3)/2}(1-u)^{-1/2}\,du,
    \label{Lin546tegral}
\end{align}
where $B(a,b)$ denotes the beta function \cite[8.17.3]{NISTHandbook} defined as
\begin{align*}
    B(a,b)=\frac{\Gamma(a)\Gamma(b)}{\Gamma(a+b)}.
\end{align*}
For $\upsilon\in\mathcal I$, let $s(\upsilon)=1-t(\upsilon)^2$. Then
\begin{align}
s(\upsilon)
&=
\frac{\delta}{\upsilon}
\left(
1-\frac{(\upsilon-1)^2}{4\delta(1-\delta)}
\right).
\label{sdi3dentity}
\end{align}
Based on the integral representation $(\ref{Lin546tegral})$ and a change of variables $u=s \tau$, we can write, for $\upsilon\in\mathcal I$,
\begin{align}
    \rho_{k,d}(\upsilon)
    &=L_k(t(\upsilon))\notag\\
    &=\lambda_k s(\upsilon)^{\frac{k-1}{2}}\ell_k(s(\upsilon)),
    \label{Ad3exact}
\end{align}
where
\begin{align}
    \lambda_k&=
    \frac{\Gamma(k/2)}
    {\sqrt{\pi}(k-1)\Gamma((k-1)/2)},
    \label{lamlamlam}\\
    \ell_k(s)&=
    \frac{k-1}{2}\int_0^1 \tau^{\frac{k-3}{2}}(1-s\tau)^{-1/2}\,d\tau.
    \label{ellnndef}
\end{align}
Since $s(\upsilon)$ stays in a compact subinterval of $(0,1)$ for $\upsilon\in\mathcal I$, the integral defining $\ell_k$ gives
\begin{align}
    1\leq\ell_k(s(\upsilon))
    \leq(1-s(\upsilon))^{-1/2}\leq C.
    \label{28uiif}
\end{align}
Hence, for $\upsilon\in\mathcal I$, we have
\begin{align}
    \rho_{k,d}(\upsilon)
    =
    \underbrace{\lambda_k}_{\sim (2\pi k)^{-1/2}}
    \cdot
    \underbrace{s(\upsilon)^{\frac{k-1}{2}}}_{\text{exponential}}
    \cdot
    \underbrace{\ell_k(s(\upsilon))}_{\text{bounded correction}}.
    \label{nmyknf}
\end{align}
Since $s(\upsilon)$ is uniformly bounded away from $0$ and below $1$ for $\upsilon\in\mathcal I$, we have
\begin{align}
    \sup_{\upsilon\in\mathcal I}\rho_{k,d}(\upsilon)
    \leq Ce^{-\gamma k}.
    \label{eq:spherical-cap-local-decay}
\end{align}
Recall the definition $W_k
\coloneqq -\log\rho_{k,d}(\Upsilon)+\xi_k(\Upsilon)$. Then
combining $(\ref{sdi3dentity})$ and $(\ref{Ad3exact})$, we obtain on the event $\mathcal G$ that
\begin{align}
W_k
&= 
-\log\lambda_k+(k-1)R(d)
+\frac{k-1}{2}h(\Upsilon)\notag\\
&\quad\quad\quad-\log\ell_k(s(\Upsilon))+\xi_k(\Upsilon),
\label{Wke343xpansion}
\end{align}
where
\begin{align}
    h(\upsilon)
    =
    \log\upsilon
    -
    \log\left(
    1-\frac{(\upsilon-1)^2}{4\delta(1-\delta)}
    \right).
    \label{hd2tef}
\end{align}
Since $h(1)=0$ and $h'(1)=1$, Taylor's theorem gives
\begin{align}
    |h(\upsilon)-(\upsilon-1)|
    \leq C(\upsilon-1)^2,
    \qquad \upsilon\in\mathcal I.
    \label{eq:Tay3glor-bound}
\end{align}
Stirling's formula yields
\begin{align}
    -\log\lambda_k
    =
    \frac12\log k+\frac12\log(2\pi)+O(k^{-1}).
    \label{stirlih5bnglambda}
\end{align}
Now let
\begin{align*}
    f_k(\upsilon)&=-\log\rho_{k,d}(\upsilon)+\xi_k(\upsilon),\\
    b_k&=kR(d)+\frac12\log k
\end{align*}
in Proposition \ref{prop:normal-approximation}. The calculation in $(\ref{Wke343xpansion})$ holds pointwise for every realization $\Upsilon = \upsilon\in\mathcal I$. Hence, for $\upsilon \in \mathcal{I}$,
\begin{align*}
    f_k(\upsilon) &= -\log\lambda_k+(k-1)R(d)
+\frac{k-1}{2}h(\upsilon)\notag\\
& \quad \quad \quad \quad-\log\ell_k(s(\upsilon))+\xi_k(\upsilon)\\
&= b_k + \frac{k-1}{2}h(\upsilon) + r_k(\upsilon), 
\end{align*}
where in the last equality above, we used $(\ref{eq:spherical-perturbation-bound})$, $(\ref{28uiif})$ and $(\ref{stirlih5bnglambda})$ so that the term $r_k(\upsilon)$ satisfies 
\begin{align}
    \sup_{k\geq\max\{k_1,2\}}
    \sup_{\upsilon\in\mathcal I}|r_k(\upsilon)|
    &\leq C.
    \label{eq:spherical-cap-local-expansion}
\end{align}
Since
\begin{align*}
    \frac{k-1}{2}h(\upsilon)-\frac{k}{2}(\upsilon-1)
    =
    \frac{k-1}{2}\bigl(h(\upsilon)-(\upsilon-1)\bigr)
    -\frac{\upsilon-1}{2},
\end{align*}
$(\ref{eq:Tay3glor-bound})$ and $(\ref{eq:spherical-cap-local-expansion})$ imply
\begin{align}
    \left|f_k(\upsilon)-b_k-\frac{k}{2}(\upsilon-1)\right|
    \leq C\bigl(1+k(\upsilon-1)^2\bigr),
    \qquad \upsilon\in\mathcal I.
    \label{eq:spherical-cap-quadratic-expansion}
\end{align}
Thus, $f_k$ satisfies the hypotheses of Proposition \ref{prop:normal-approximation}, with constants depending only on $\delta$, $\eta_1$, $k_1$ and $C_1$. Applying the proposition proves $(\ref{eq:spherical-cap-normal})$.

\section{Proof of Theorem \ref{thm_achievability} \label{thm_achievability_proof}}

Let $\widehat Z^k$ be the output of the decoder when $S^k$ is the input to the encoder. We start by considering the conditional probability 
\begin{align}
    \Pr\left (\|S^k-\widehat Z^k\|^2 \leq kd \big | S^k = s^k, E_\mathcal{J}\right), \label{condsuccess}
\end{align}
which denotes the probability
conditioned on the source realization $s^k$ and on the event $E_\mathcal{J}$ that the channel output $Y^n$ has $n - j$ erasures, and the unerased block positions are
$\mathcal J\subseteq\{1,\ldots,n\}$ where $|\mathcal J|=j$. Let $\upsilon = \|s^k\|^2/(k \sigma^2)$.

The decoder in $(\ref{bec_dec})$ checks whether for each entry  $(\bar{X}_i^n, \bar{Z}_i^k)$ in the codebook, $\bar{X}_i^n$ agrees with $Y^n$ in every unerased position in $\mathcal{J}$. Such entries will be called decoder-consistent. The decoder then outputs the $\bar{Z}_i^k$ from the first such entry. On the other hand, the encoder in $(\ref{bec_enc})$ transmits the $\bar{X}_i^n$ from the first entry $(\bar{X}_i^n, \bar{Z}_i^k)$ whose reproduction component $\bar{Z}_i^k$ lies in $B_d(s^k)$. Let $I$ and $D$
denote the encoder-selected and decoder-selected indices, respectively. 
On the event $E_{\mathcal J}$, the received unerased symbols at the decoder are
\[
(Y_\ell)_{\ell\in\mathcal J}
=
(\bar X_{I,\ell})_{\ell\in\mathcal J}.
\]
Consequently,
\[
D
=
\min\left\{
i\in\mathbb N:
(\bar X_{i,\ell})_{\ell\in\mathcal J}
=
(\bar X_{I,\ell})_{\ell\in\mathcal J}
\right\}.
\]
In particular, $D\leq I$ and $D < \infty$ almost surely.

Recall that $\rho_{k,d}(\upsilon) = P_{Z^k}(B_d(s^k))$. 
We first handle the case $\rho_{k,d}(\upsilon)=0$. In this case, the encoder definition gives
$I=1$. Since index $1$ is decoder-consistent and no smaller index exists,
we also have $D=1$. Hence, $\widehat Z^k=\bar Z_1^k$. It follows that the conditional
excess-distortion probability is one. This agrees with
\[
\frac{1-\rho_{k,d}(\upsilon)}
{1+(2^{mj}-1)\rho_{k,d}(\upsilon)}
=
1.
\]

We now assume that $\rho_{k,d}(\upsilon)>0$. Because the reproduction codewords are
independent and each one belongs to $B_d(s^k)$ with probability $\rho_{k,d}(\upsilon)$, the
encoder-selected index has the geometric distribution
\begin{align}
\Pr\left(
I=t
\,\middle|\,
S^k=s^k,E_{\mathcal J}
\right)
=
(1-\rho_{k,d}(\upsilon))^{t-1}\rho_{k,d}(\upsilon),
\qquad t\in\mathbb N.
\label{eq:encoder-index-geometric}
\end{align}
Note that the event $\{I = t \}$ above is independent of the event $E_{\mathcal{J}}$ since the latter concerns the random erasures introduced by the channel on to the channel input, and the reproduction codewords are independently generated from the channel input codewords.

Conditioned on $I=t$, the reproduction codewords satisfy $\bar Z_i^k\notin B_d(s^k)$ for $i = 1, \ldots, t - 1$, and $\bar Z_t^k\in B_d(s^k).$ Since the index $I$ is determined only by the reproduction codewords, the
channel-input codewords $\bar X_1^n,\ldots,\bar X_t^n$ remain independent and uniformly distributed over $\mathcal X^n$
conditional on $I=t$. Conditioned further on
$(\bar X_{t,\ell})_{\ell\in\mathcal J}$, each preceding channel-input
codeword agrees with $\bar X_t^n$ in every position in $\mathcal J$ with
probability $2^{-mj}$, independently of the other preceding channel-input
codewords. This is because $2^{-mj}$ is the probability that two independent channel-input sequences,
each uniformly distributed over $\mathcal X^n$, agree in all the positions
in $\mathcal J$. It then follows that
\begin{align}
\Pr\left(
D=t
\,\middle|\,
S^k=s^k,E_{\mathcal J},I=t
\right)
=
(1-2^{-mj})^{t-1}.
\label{eq:decoder-chooses-encoder-index}
\end{align}

Conditional on $I=t$, all reproduction codewords with indices less than
$t$ lie outside $B_d(s^k)$, while the encoder-selected reproduction
$\bar Z_t^k$ lies in $B_d(s^k)$. Therefore, successful reconstruction
occurs if and only if the decoder selects index $t$. We thus obtain
\begin{align*}
&\Pr\left(
\|S^k-\widehat Z^k\|^2\leq kd
\,\middle|\,
S^k=s^k,E_{\mathcal J},I=t
\right)\\
&\qquad =
(1-2^{-mj})^{t-1}.
\end{align*}
For brevity, let $r = \rho_{k,d}(\upsilon)$. Then averaging over $I$ using
\eqref{eq:encoder-index-geometric} gives
\begin{align*}
&\Pr\left(
\|S^k-\widehat Z^k\|^2\leq kd
\,\middle|\,
S^k=s^k,E_{\mathcal J}
\right)\\
&\qquad =
\sum_{t=1}^{\infty}
(1-r)^{t-1}r\,
(1-2^{-mj})^{t-1}\\
&\qquad = \frac{2^{mj}r}
{1+\bigl(2^{mj}-1\bigr)r}.
\end{align*}
Hence, 
\begin{align}
&\Pr\left(
\|S^k-\widehat Z^k\|^2>kd
\,\middle|\,
S^k=s^k,E_{\mathcal J}
\right)\notag\\
&\qquad =
1-
\frac{2^{mj}\rho_{k,d}(\upsilon)}
{1+\bigl(2^{mj}-1\bigr)\rho_{k,d}(\upsilon)}\notag\\
&\qquad =
\frac{1-\rho_{k,d}(\upsilon)}
{1+\bigl(2^{mj}-1\bigr)\rho_{k,d}(\upsilon)}.
\label{eq:conditional-error}
\end{align}
As shown above, this formula also holds when $\rho_{k,d}(\upsilon)=0$. Thus, for every source realization
$s^k\neq 0$ and every erasure pattern $\mathcal J$ of cardinality $j$, we have
\begin{align}
\Pr\left(
\|S^k-\widehat Z^k\|^2>kd
\,\middle|\,
S^k=s^k,E_{\mathcal J}
\right)
=
\Psi_j(\upsilon;d). \label{8hk2jtl}    
\end{align}
Note that the RHS depends on $E_{\mathcal{J}}$ only through $|\mathcal{J}| = j$, not the positions of the unerased blocks, and it depends on $s^k$ only through the normalized energy $\upsilon = \|s^k\|^2/(k \sigma^2)$. We have $\operatorname{Pr}(|\mathcal{J}| = j) = \pi_j$ as given in \eqref{njerasureslej4}. Averaging \eqref{8hk2jtl} over $\Upsilon$ and $|\mathcal{J}|$, noting that $\Upsilon \ind |\mathcal{J}|$, we obtain 
\begin{align*}
\Pr\left (\|S^k-\widehat Z^k\|^2 > kd \right)
&=
\sum_{j=0}^n
\pi_j
\int_0^\infty
p_{\Upsilon}(\upsilon)\Psi_j(\upsilon;d)\,d\upsilon.
\end{align*}
The expression above is the excess-distortion probability averaged over
the random codebook. Therefore, there exists a realization of the random
codebook whose excess-distortion probability is no larger than this
average. Hence, whenever
\[
\epsilon
\geq
\sum_{j=0}^n
\pi_j
\int_0^\infty
p_{\Upsilon}(\upsilon)\Psi_j(\upsilon;d)\,d\upsilon,
\]
there exists a $(d,\epsilon,n,k)$ source--channel code.

\section{Proof of Theorem \ref{asymp_thm_achievability} \label{asymp_thm_achievability_proof} }

Let $\rho_{k, d}(\upsilon) = L_k(t(\upsilon))$, where 
\begin{align*}
     t(\upsilon) &= \frac{\upsilon + 1 - 2\delta}{2\sqrt{\upsilon} \sqrt{1-\delta}}. 
\end{align*}
Then note that
\begin{align*}
    \frac{1-\rho_{k, d}(\upsilon)}
{1+\bigl(2^{mj}-1\bigr)\rho_{k, d}(\upsilon)} = \frac{1}
{1+2^{mj} \frac{\rho_{k, d}(\upsilon)}{1 - \rho_{k, d}(\upsilon) }},
\end{align*}
where we use the convention $1/0 = \infty$ and $1/\infty = 0$ for the RHS above. 

Let $P_e(d,n,k)$ denote the ensemble-average excess-distortion
probability of the code $C_{n,k}^{(d)}$ in \eqref{bec_enc} - \eqref{bec_dec}.
In the proof of  Theorem~\ref{thm_achievability}, we showed that it is equal to 
\begin{align*}
\mathbb{E}\left [ \frac{1}{1 + e^{J_n - W_k}}  \right],
\end{align*}
where
\begin{align*}
    J_n &= \log(2^m) \sum_{i=1}^n B_i,\\
    \Upsilon &\stackrel{d}{=} \frac{\chi_k^2}{k},\\
    W_k &= - \log \left(\rho_{k, d}(\Upsilon)\right) + \log(1 -\rho_{k, d}(\Upsilon) ), 
\end{align*}
$B_1, \ldots, B_n$ are independent, each $B_i \sim \operatorname{Bern}(1-p_i)$, and $J_n \ind \Upsilon$. In the definition of $W_k$, we use the convention $-\log 0 = +\infty$. Let $U$ be an independent random variable with CDF given by
\begin{align}
    F_U(u) = \frac{1}{1 + e^{-u}}. \label{32095} 
\end{align}
Then 
\begin{align*}
    \mathbb{P} \left [W_k + U \geq J_n \big | J_n, \Upsilon \right ] = \frac{1}{1 + e^{J_n - W_k}}.
\end{align*}
Hence, 
\begin{align}
    \mathbb{E}\left [ \frac{1}{1 + e^{J_n - W_k}}  \right] = \mathbb{P} \left [W_k + U \geq J_n \right ], \label{marium1}
\end{align}
where $U \ind W_k \ind J_n$. 

We apply Corollary \ref{coro_specialize} with
\[
\xi_k(\upsilon)=\log(1-\rho_{k,d}(\upsilon)).
\]
But first, we must check that the hypothesis \eqref{eq:spherical-perturbation-bound} in Corollary \ref{coro_specialize} holds. The factorization \eqref{nmyknf} of $\rho_{k, d}(\upsilon)$ and the bound \eqref{eq:spherical-cap-local-decay} in Appendix~\ref{prop:normal-approximation_proof}, which hold on a fixed interval $I=[1-\eta,1+\eta]$, where $\eta>0$ depends only on $\delta$, give us that $\rho_{k,d}(\upsilon)\leq Ce^{-\gamma k}$ uniformly on $I$. Here and below, the constants in this source approximation depend only on $\delta$. For all sufficiently large $k$, we therefore have $\rho_{k,d}(\upsilon)\leq1/2$ on $I$, and hence
\begin{align*}
|\xi_k(\upsilon)|
&=-\log(1-\rho_{k,d}(\upsilon))
\leq2\rho_{k,d}(\upsilon)
\leq Ce^{-\gamma k},
\qquad \upsilon\in I.
\end{align*}
Thus $\xi_k$ is uniformly bounded on this fixed neighborhood and satisfies \eqref{eq:spherical-perturbation-bound}. Define
\begin{align*}
\widehat W_k
&=\frac{W_k-kR(d)-\frac12\log k}{\sqrt{k/2}}.
\end{align*}
By \eqref{eq:spherical-cap-normal}, for all sufficiently large $k$,
\begin{align*}
\sup_{w\in\mathbb R}
\left|\mathbb P[\widehat W_k\leq w]-\Phi(w)\right|
\leq\frac{C}{\sqrt{k}}.
\end{align*}
Conditioning on the independent logistic random variable $U$, we obtain
\begin{align}
    &\left|
\mathbb{P}\left[
\frac{
W_k+U-kR(d)-\frac12\log k
}{
\sqrt{k/2}
}
\leq x
\right]
-\Phi(x)
\right| \label{09j,k}\\
&= \left|
\mathbb{E}\left [\mathbb{P}\left[
\frac{
W_k+U-kR(d)-\frac12\log k
}{
\sqrt{k/2}
}
\leq x \Big | U
\right] \right] 
-\Phi(x)
\right|\\
&= \left|
\mathbb{E}\left [\mathbb{P}\left[
\frac{
W_k+U-kR(d)-\frac12\log k
}{
\sqrt{k/2}
}
\leq x \Big | U
\right] -\Phi(x) \right] 
\right|\\
&= \left|
\mathbb{E}\left [\mathbb{P}\left[\widehat W_k
\leq x - \frac{U}{\sqrt{k/2}} \Big | U
\right] -\Phi(x) \right] 
\right|\\
&= \left|
\mathbb{E}\left [\mathbb{P}\left[\widehat W_k
\leq x - \frac{U}{\sqrt{k/2}} \Big | U
\right] - \Phi\left(x - \frac{U}{\sqrt{k/2}} \right) + \Phi\left(x - \frac{U}{\sqrt{k/2}} \right) -  \Phi(x) \right] 
\right|\\
&\leq
\frac{C}{\sqrt{k}}
+
\mathbb{E}\left[
\left|
\Phi\left(x-\frac{U}{\sqrt{k/2}}\right)-\Phi(x)
\right|
\right]\\
&\leq
\frac{C}{\sqrt{k}}
+
\frac{\mathbb{E}\left[|U|\right]}{\sqrt{\pi k}}.
\end{align}
It can be checked that $\mathbb{E}\left[|U|\right] = 2 \log(2)$. Hence, 
\begin{align}
\sup_{x\in\mathbb{R}}
\left|
\mathbb{P}\left[
\frac{
W_k+U-kR(d)-\frac12\log k
}{
\sqrt{k/2}
}
\leq x
\right]
-\Phi(x)
\right|
\leq
\frac{C}{\sqrt{k}}.
\label{eq:W-plus-U-normal}
\end{align}
Since \(U\) has a continuous distribution, \(W_k+U\) has no atoms at
finite points. Hence, \eqref{eq:W-plus-U-normal} implies that for every \(j\in\mathbb{R}\),
\begin{align}
\mathbb{P}[W_k+U\geq j]
\leq
Q\left(
\frac{
j-kR(d)-\frac12\log k
}{
\sqrt{k/2}
}
\right)
+
\frac{C}{\sqrt{k}}.
\label{eq:source-tail-bound}
\end{align}
Since \(J_n\) is independent of \(W_k+U\), conditioning on \(J_n\)
and applying \((\ref{eq:source-tail-bound})\), we obtain
\begin{align}
\mathbb{E}\left [ \frac{1}{1 + e^{J_n - W_k}}  \right]
&=
\mathbb{P}[W_k+U\geq J_n] \label{rer11}\\
&\leq
\mathbb{E}\left[
Q\left(
\frac{
J_n-kR(d)-\frac12\log k
}{
\sqrt{k/2}
}
\right)
\right]
+
\frac{C}{\sqrt{k}}. \label{rer22}
\end{align}
This proves Proposition \ref{semiasympt}.

The rest of the proof uses Proposition \ref{semiasympt} as the starting point. From Proposition \ref{semiasympt}, there exists an $(n, k)$ source--channel code for sufficiently large $k$ whose excess-distortion probability $P_e(d, n, k)$ is upper bounded by 
\begin{align*}
    &\frac{K}{\sqrt{k}} +  \sum_{j=0}^n \pi_j \, Q \left( \frac{j m \log(2) - k R(d) - \frac{1}{2} \log(k) }{\sqrt{k/2}} \right)\\
    &= \mathbb{P} \left( \sqrt{\frac{k}{2}} Z \geq J_n - k R(d) - \frac{1}{2} \log(k)  \right) + \frac{K}{\sqrt{k}},
\end{align*}
where $Z \ind J_n$, $Z \sim \mathcal{N}(0, 1)$ and $J_n = m \log(2) \sum_{i=1}^n B_i$, where $B_1, \ldots, B_n$ are independent and $B_i \sim \operatorname{Bern}(1 - p_i)$.

The constants in Proposition \ref{semiasympt} are uniform in $n$ and in
$p_1,\ldots,p_n$. Indeed, the approximation used in that proposition is
applied only to the source random variable (via Corollary \ref{coro_specialize}), while the distribution of
$J_n$ is retained exactly.

The mean and variance of $J_n$ are
\begin{align}
\mu_n
&=
\log(2^m)\sum_{i=1}^n(1-p_i),\\
V_n
&=
\bigl(\log(2^m)\bigr)^2
\sum_{i=1}^n p_i(1-p_i).
\end{align}
Also, 
\begin{align}
    S_{n,k}^2
&=
\frac{k}{2}+V_n = \operatorname{Var}\left(\sqrt{\frac{k}{2}}Z-(J_n-\mu_n) \right). \label{2398yt34u} 
\end{align}
Consequently,
\begin{align}
P_e(d, n, k)
&\leq
\mathbb{P}\left(
\sqrt{\frac{k}{2}}Z-(J_n-\mu_n)
\geq
\mu_n-kR(d)-\frac{1}{2}\log k
\right)
+
\frac{K}{\sqrt{k}}.
\label{eq:Pe-centered}
\end{align}

Let $Z_1,\ldots,Z_r$ be independent standard Gaussian random variables,
independent of $J_n$. Since Gaussian variances add,
$\sqrt{k/2}Z$ has the same distribution as
$\sum_{\ell=1}^r\sqrt{k/(2r)}Z_\ell$. We may therefore apply the
Berry--Esseen theorem to the $n+r$ independent summands
$-m\log(2)(B_i-(1-p_i))$, $1\leq i\leq n$, and
$\sqrt{k/(2r)}Z_\ell$, $1\leq\ell\leq r$. The total contribution of
the Gaussian summands to the sum of third absolute moments is
$r(k/(2r))^{3/2}\mathbb{E}[|Z|^3]
=(k/2)^{3/2}\mathbb{E}[|Z|^3]/\sqrt{r}$, which goes to zero as
$r\to\infty$. The distribution and variance of the total sum do not
depend on $r$. Since the resulting Berry--Esseen inequality holds for
every $r$, taking $r\to\infty$ removes the Gaussian third-moment term,
leaving only the third absolute moments of the Bernoulli summands. For the Bernoulli summands, 
\begin{align}
\mathbb{E}\left[
\left|
m\log(2)\bigl(B_i-(1-p_i)\bigr)
\right|^3
\right]
&=
\bigl(m\log(2)\bigr)^3
p_i(1-p_i)
\left(
p_i^2+(1-p_i)^2
\right)
\nonumber\\
&\leq
\bigl(m\log(2)\bigr)^3p_i(1-p_i).
\label{eq:third-moment}
\end{align}
It follows that
\begin{align}
P_e(d, n, k)
&\leq
Q\left(
\frac{
\mu_n-kR(d)-\frac{1}{2}\log k
}{
S_{n,k}
}
\right)
+
C_{\mathrm{BE}}
\frac{m\log(2)V_n}{S_{n,k}^3}
+
\frac{K}{\sqrt{k}},
\label{eq:Pe-BE}
\end{align}
where $C_{\mathrm{BE}}$ is an absolute constant.

Since
\begin{align}
\sup_{v\geq 0}
\frac{v}{(k/2+v)^{3/2}}
&=
\frac{2\sqrt{2}}{3\sqrt{3}}\frac{1}{\sqrt{k}},
\label{eq:variance-remainder}
\end{align}
there is a constant $D$, independent of $n$, $k$, and
$p_1,\ldots,p_n$, such that
\begin{align}
P_e(d, n, k)
&\leq
Q\left(
\frac{
\mu_n-kR(d)-\frac{1}{2}\log k
}{
S_{n,k}
}
\right)
+
\frac{D}{\sqrt{k}}.
\label{eq:Pe-normal}
\end{align}

Let $q = Q^{-1}(\epsilon)$ for some fixed $\epsilon \in (0, 1)$ and 
\begin{align}
x
&=
\frac{
\mu_n-kR(d)-\frac{1}{2}\log k
}{
S_{n,k}
}.
\label{eq:x-definition}
\end{align}

If $x\geq q+1$, then
\begin{align}
Q(x)+\frac{D}{\sqrt{k}}
&\leq
Q(q+1)+\frac{D}{\sqrt{k}}
\leq
Q(q)
=
\epsilon
\label{eq:large-x-case}
\end{align}
for all sufficiently large $k$, which implies that $P_e(d, n, k) \leq \epsilon$. 

If $x<q+1$, then we obtain from 
$(\ref{eq:x-definition})$ that, for constants $A > 0$ and $B$
depending only on the fixed parameters,
\begin{align}
\mu_n
&\leq
Ak+BS_{n,k}.
\label{eq:mu-upper-bound}
\end{align}
Moreover, the Bernoulli channel satisfies
\begin{align}
V_n
&=
\bigl(m\log(2)\bigr)^2
\sum_{i=1}^n p_i(1-p_i)
\nonumber\\
&\leq
\bigl(m\log(2)\bigr)^2
\sum_{i=1}^n(1-p_i)
=
m\log(2)\mu_n.
\label{eq:variance-mean}
\end{align}
Combining $(\ref{2398yt34u})$, $(\ref{eq:mu-upper-bound})$ and $(\ref{eq:variance-mean})$ gives
\begin{align}
S_{n,k}^2
&\leq
A'k+B'S_{n,k}
\label{eq:S-quadratic}
\end{align}
for suitable constants $A' > 0$ and $B'$. Solving
$(\ref{eq:S-quadratic})$, we obtain a constant $C > 0$ such that
\begin{align}
S_{n,k}
&\leq
C\sqrt{k}.
\label{eq:S-bound}
\end{align}

Let $\phi$ denote the standard normal PDF and set
$\phi_*=\min_{u\in[q,q+1]}\phi(u)>0$. Let the constant $L = \frac{CD}{\phi_*}$. Then for sufficiently large $k$, 
\begin{align}
    q + \frac{L}{C \sqrt{k}} < q + 1. 
\end{align}
Let $q + \frac{L}{C \sqrt{k}} \leq x < q + 1$. 
Then 
\begin{align}
Q(q)-Q(x)
&=
\int_q^x\phi(u)\,du
\nonumber\\
&\geq
\phi_*(x-q)
\nonumber\\
&\geq
\frac{\phi_*L}{C\sqrt{k}} = \frac{D}{\sqrt{k}}.
\label{eq:normal-tail-gap}
\end{align}
We then obtain from $(\ref{eq:Pe-normal})$ and
$(\ref{eq:normal-tail-gap})$ that
\begin{align}
P_e(d, n, k)
&\leq
Q(x)+\frac{D}{\sqrt{k}}
\leq
Q(q)
=
\epsilon
\label{eq:final-error-bound}
\end{align}
for all $q + 1 > x \geq q + \frac{L}{C \sqrt{k}}$.

To conclude, denote the constant in
\eqref{eq:gaussian-corollary-statement} by $C_{\mathrm{th}}$
and choose $C_{\mathrm{th}}\geq L$. Then \eqref{eq:gaussian-corollary-statement} implies that 
\begin{align}
x\geq q+\frac{C_{\mathrm{th}}}{S_{n,k}}. \label{823rg8u4girr}    
\end{align}
If $x<q+1$, then \eqref{eq:S-bound} gives
\[
x\geq q+\frac{C_{\mathrm{th}}}{C\sqrt{k}}
\geq q+\frac{L}{C\sqrt{k}},
\]
so \eqref{eq:final-error-bound} holds.
If $x\geq q+1$, \eqref{eq:large-x-case} holds. Increasing $k_0$ to satisfy the preceding lower bounds on $k$
proves the theorem. Both $C_{\mathrm{th}}$ and $k_0$ depend
only on $\delta$, $\epsilon$, and $m$. Hence, we have proved that \eqref{eq:gaussian-corollary-statement}, which implies \eqref{823rg8u4girr}, is a sufficient condition for an optimal source--channel code to be $(d, \epsilon, n, k)$.

\section{Proof of Theorem \ref{convo2} \label{convo2_proof}}

Let $\Upsilon=\|S^k\|^2/(k\sigma^2)$. Conditional on $\Upsilon=\upsilon>0$, the source is uniform on the sphere of radius $r=\sqrt{k\sigma^2\upsilon}$. We first upper bound the maximum fraction of this sphere that can be covered by a ball of radius $\sqrt{kd}$ centered at any arbitrary $z^k\in\mathbb R^k$. This upper bound will be uniform over all $z^k$. For $\upsilon \leq \delta$, we use the trivial upper bound of $1$. Suppose $\upsilon>\delta$ and let $a=\|z^k\|$. If $z^k = \mathbf{0}$, none of the sphere is covered because when $\upsilon > \delta$, we have $r > \sqrt{kd}$. For $a>0$, Corollary \ref{lemmaball_coro} gives us
\begin{align}
\mathbb P\left[\|S^k-z^k\|^2\leq kd\,\middle|\,\Upsilon=\upsilon\right]
&=L_k\left(\frac{r^2+a^2-kd}{2ra}\right).
\label{eq:source-sphere-cap}
\end{align}
Since $r^2>kd$, we have
\begin{align}
\frac{r^2+a^2-kd}{2ra}
&= \frac{1}{2r}\left(a+\frac{r^2-kd}{a}\right)
\geq\frac{\sqrt{r^2-kd}}{r}
=\sqrt{1-\frac{\delta}{\upsilon}},
\label{eq:optimal-cap-angle}
\end{align}
with equality at $a=\sqrt{r^2-kd}$. Since $L_k$ is nonincreasing, the maximum conditional covering probability is exactly $\alpha_{k,d}(\upsilon)$ in $(\ref{eq:maximal-source-cap})$. Specifically, 
\begin{equation}
    \sup_{z^k \in \mathbb{R}^k} \mathbb P\left[\|S^k-z^k\|^2\leq kd\,\middle|\,\Upsilon=\upsilon\right] \leq \alpha_{k,d}(\upsilon). \label{3ujbk}
\end{equation}

Now consider a deterministic source--channel code. Let $\mathcal J\subseteq\{1,\ldots,n\}$ denote the set of unerased block positions, and fix any realization $\mathcal{J} = \mathcal{J}_{r}$ with positive probability and $|\mathcal{J}_{r}|=j$. Since the erasures are independent of the channel input, $\mathcal J$ is independent of $S^k$. Conditioned on $\mathcal J = \mathcal{J}_{r}$, there are exactly $2^{mj}$ possible channel output sequences. Consequently, the set $\mathcal R_{\mathcal{J}_{r}}$ of reproduction sequences that a deterministic decoder, denoted as $\operatorname{g}$, can output satisfies
\begin{align}
|\mathcal R_{\mathcal{J}_{r}}|\leq 2^{mj}.
\label{eq:conditional-reproduction-count}
\end{align}
Successful reproduction requires that
\begin{align}
S^k\in A_{\mathcal{J}_{r}} \coloneqq \bigcup_{z^k\in\mathcal R_{\mathcal{J}_{r}}} B_d(z^k).
\label{eq:conditional-success-union}
\end{align}

Applying the union bound to $(\ref{eq:conditional-success-union})$ and using $(\ref{3ujbk})$, we obtain
\begin{align}
\mathbb P\left[\|S^k-\operatorname g(Y^n)\|^2>kd\,\middle|\,\Upsilon=\upsilon,\mathcal{J}_{r}\right]
&\geq\max\{0,1-2^{mj}\alpha_{k,d}(\upsilon)\}.
\end{align}
Averaging over $\Upsilon$, whose distribution remains $\chi_k^2/k$ under conditioning on $\mathcal{J}_{r}$, gives
\begin{align}
\mathbb P\left[\|S^k-\operatorname g(Y^n)\|^2>kd\,\middle|\,\mathcal{J}_{r}\right]
&\geq\int_0^\infty p_{\Upsilon}(\upsilon)\max\{0,1-2^{mj}\alpha_{k,d}(\upsilon)\}\,d\upsilon.
\label{eq:cap-converse-conditional}
\end{align}

We next obtain a second lower bound for the LHS of \eqref{eq:cap-converse-conditional}. Recall from \eqref{eq:conditional-success-union} that 
\begin{align}
    \operatorname{Pr} \left( \|S^k - \operatorname{g}(Y^n) \|^2 \leq kd \,\middle|\,\mathcal J_r  \right) &\leq \operatorname{Pr}\left( S^k\in A_{\mathcal{J}_r} \right). \label{3498g}
\end{align}
Let
\begin{align}
v_{k,d}\coloneqq \frac{\pi^{k/2}(kd)^{k/2}}{\Gamma(k/2+1)} \label{9bb}
\end{align}
be the volume of the distortion ball of radius $\sqrt{kd}$. The union in $(\ref{eq:conditional-success-union})$ has volume at most $2^{mj}v_{k,d}$. Since the Gaussian density is a decreasing function of the distance from the origin, among measurable sets with volume at most $2^{mj}v_{k,d}$, the zero-centered ball of this volume has the largest Gaussian probability. Its radius is $\sqrt{kd}\,2^{mj/k}$, and we denote this ball by $\mathcal{B}(0, \sqrt{kd}\,2^{mj/k})$. 

Hence, we can upper bound \eqref{3498g} by 
\begin{align*}
    &\operatorname{Pr}\left( S^k\in \mathcal{B}(0, \sqrt{kd}\,2^{mj/k}) \right)\\
    &= \operatorname{Pr}\left( \|S^k \|^2 \leq kd 2^{2mj/k} \right)\\
    &= F_{\chi_k^2} \left(k\delta 2^{2mj/k} \right)
\end{align*}
Hence, 
\begin{align}
\mathbb P\left[\|S^k-\operatorname g(Y^n)\|^2>kd \,\middle|\,\mathcal J_r\right]
&\geq 1-F_{\chi_k^2}\left(k\delta\,2^{2mj/k}\right).
\label{eq:volume-converse-conditional}
\end{align}
Taking the maximum of $(\ref{eq:cap-converse-conditional})$ and $(\ref{eq:volume-converse-conditional})$, and averaging over $\mathcal J$, gives
\begin{align}
\mathbb P\left[\|S^k-\operatorname g(Y^n)\|^2>kd\right]
&\geq\sum_{j=0}^n\pi_j\Theta_j(k,d),
\label{eq:geometric-converse-error}
\end{align}
because $\mathbb P[|\mathcal J|=j]=\pi_j$. For a randomized code, condition first on all source-independent random seeds used by the encoder and decoder. The resulting deterministic code satisfies $(\ref{eq:geometric-converse-error})$ for every realization of these seeds, with the same right-hand side. Averaging over the seeds proves the same inequality for randomized codes. A $(d,\epsilon,n,k)$ code has excess-distortion probability at most $\epsilon$, so $(\ref{eq:geometric-converse-error})$ implies $(\ref{m-})$.

\subsection{Comparison with the list-code converse \label{list_code_converse}}

Let $\lambda$ be the Lebesgue measure on $\mathbb R^k$. Following
\cite{kostina_JSCC}, an $(\epsilon,L,\lambda)$ list code has a decoder
whose output is a measurable set $\mathcal L\subseteq\mathbb R^k$
satisfying $\lambda(\mathcal L)\leq L$ almost surely and
$\mathbb P[S^k\notin\mathcal L]\leq\epsilon$. Thus, the size of the
list is measured by its Lebesgue volume. Given a
$(d,\epsilon,n,k)$ source--channel code, map its decoder output $Z^k$
to the list $\mathcal L=B_d(Z^k)$. This list has volume $v_{k,d}$
given in \eqref{9bb}, and
\begin{align*}
S^k\notin\mathcal L
\iff \|S^k-Z^k\|^2>kd.
\end{align*}
The LHS is the list-decoding error event for the list code while the
RHS is the excess-distortion event for the source--channel code.
Hence, every $(d,\epsilon,n,k)$ source--channel code induces an
$(\epsilon,v_{k,d},\lambda)$ list code, and any necessary condition
(i.e., a converse result) for the latter also applies to the former.

The list-code converse in \cite[Theorem~5]{kostina_JSCC} is based on binary
hypothesis testing. Let $P$ be a probability measure and $Q$ a
$\sigma$-finite measure on a common measurable space $\Omega$.
A possibly randomized test $T:\Omega\to\{0,1\}$ returns a binary
decision.
Write $P_T(t\mid\omega)$ for its conditional probability of decision
$t$ given observation $\omega$, with decision $1$ meaning that it
chooses $P$. For $\alpha\in[0,1]$, define
\begin{align}
\beta_\alpha(P,Q)
&\coloneqq
\inf_{\substack{P_T:\\
 \int P_T(1\mid\omega)\,P(d\omega)\geq\alpha}}
\int P_T(1\mid\omega)\,Q(d\omega).
\label{30tuk,w}
\end{align}
The infimum is over these stochastic conditional laws $P_T$.
The constraint requires the probability of decision $1$ under
$P$ to be at least $\alpha$. When $Q$ is a probability measure, the
objective is the probability of incorrectly choosing $P$ under $Q$.
For a general $\sigma$-finite $Q$, the objective is instead the
$Q$-measure assigned to decision $1$ and need not be a probability.
For a deterministic test, $P_T(1\mid\omega)=T(\omega)$, so the two
integrals reduce to $\int T\,dP$ and $\int T\,dQ$.

\textit{Specializing the hypothesis-testing condition:}
In \cite[Theorem~5]{kostina_JSCC}, $Q_S$ is the measure used to define
the size of the decoder's output list, and $L$ is the upper bound on
that size. In that theorem, we take the source to be $S^k$, the
channel to be $P_{Y^n|X^n}$, the list measure to be $Q_S=\lambda$, the list size to be $L=v_{k,d}$, and the
auxiliary output probability measure $P_{\overline Y}$ to be the following product distribution:
\begin{align}
P_{\overline Y^n}
&=\bigotimes_{i=1}^nP_{\overline Y_i},
\qquad
P_{\overline Y_i}(y_i)=
\begin{cases}
p_i, & y_i=e,\\
(1-p_i)2^{-m}, & y_i\in\mathcal X.
\end{cases}
\label{eq:list-comparison-auxiliary}
\end{align}
We also fix any reference input $x^n\in\mathcal X^n$. To verify the symmetry
assumption of that theorem, write
$\mathcal J(y^n)=\{i:y_i\neq e\}$ and
$j(y^n)=|\mathcal J(y^n)|$. For every $y^n$ with
$P_{\overline Y^n}(y^n)>0$,
\begin{align}
\log\frac{P_{Y^n|X^n=x^n}(y^n)}{P_{\overline Y^n}(y^n)}
&=
\begin{cases}
mj(y^n)\log 2,
& y_i=x_i\text{ for every }i\in\mathcal J(y^n),\\
-\infty, & \text{otherwise.}
\end{cases}
\label{eq:list-comparison-symmetry}
\end{align}
Under $P_{Y^n|X^n=x^n}$, every unerased block agrees with $x^n$, and
the number of unerased blocks has PMF $\pi_j$. Under
$P_{\overline Y^n}$, the erasure pattern has the same distribution.
Conditional on any erasure pattern of positive probability with $j$
unerased blocks, their values are uniform on $\mathcal X^j$, so they
agree with $x^n$ with probability $2^{-mj}$. Therefore, the
distribution of the log-likelihood ratio in
\eqref{eq:list-comparison-symmetry} does not depend on $x^n$ under
either output law, as required by the theorem.

Define two measures on $\Omega=\mathbb R^k\times\mathcal Y^n$ by
\begin{align}
P_{x^n}
&=P_{S^k}\otimes P_{Y^n|X^n=x^n},
\qquad
Q=\lambda\otimes P_{\overline Y^n}.
\label{eq:list-comparison-measures}
\end{align}
The measure $P_{x^n}$ is a probability measure under which the
Gaussian source is independent of a channel output generated using
the fixed input $x^n$. The measure $Q$ is $\sigma$-finite because its
source factor is Lebesgue measure. With the choices above,
\cite[Theorem~5]{kostina_JSCC}
requires every $(\epsilon,v_{k,d},\lambda)$ list code, and hence every
$(d,\epsilon,n,k)$ source--channel code, to satisfy
\begin{align}
\beta_{1-\epsilon}(P_{x^n},Q)\leq v_{k,d}.
\label{27rvbf}
\end{align}

We now show that every parameter choice satisfying \eqref{m-} also
satisfies \eqref{27rvbf}, establishing that our necessary condition
is at least as strong as this specialization of the list-code
converse. 

First, note that if there exists a deterministic test
$T:\Omega\to\{0,1\}$ satisfying
\begin{align}
\int T\,dP_{x^n}\geq1-\epsilon, \qquad \int T\,dQ\leq v_{k,d},
\label{14ses}
\end{align}
then \eqref{27rvbf} holds. Therefore, it suffices to show that if \eqref{m-} holds, there indeed exists such a test.

\textit{Constructing a test and evaluating its two integrals:} This test observes a pair $(s^k,y^n)$ and returns decision $1$ to
choose $P_{x^n}$. The
first integral in \eqref{14ses} is the probability of decision $1$ under $P_{x^n}$.
The second integral is the Lebesgue volume
of the accepted source set for each $y^n$, averaged using
$P_{\overline Y^n}(y^n)$.

For $j\in\{0,\ldots,n\}$, let $r_j=\sqrt{kd}\,2^{mj/k}.$ Define the deterministic test
\begin{align*}
T_{x^n}(s^k,y^n)
&=\mathds{1}\{y_i=x_i\text{ for every }i\in\mathcal J(y^n)\}\times
\mathds{1}\{\|s^k\|\leq r_{j(y^n)}\}.
\end{align*}
Thus, the test returns $1$ when the unerased output blocks agree with
the reference input and the source vector belongs to the indicated
centered ball. The radius $r_j$ is chosen so that this ball has
volume $2^{mj}v_{k,d}$.

First evaluate the integral under $P_{x^n}$. Agreement with $x^n$
holds automatically under this measure, the number of unerased
blocks has PMF $\pi_j$, and $S^k$ is independent of the channel
output. Since $\|S^k\|^2/\sigma^2$ has the $\chi_k^2$ distribution,
\begin{align}
\int T_{x^n}\,dP_{x^n}
&=\sum_{j=0}^n\pi_j\,\mathbb P[\|S^k\|\leq r_j]\notag\\
&=\sum_{j=0}^n\pi_j
F_{\chi_k^2}\left(\frac{r_j^2}{\sigma^2}\right)\notag\\
&=\sum_{j=0}^n\pi_j
F_{\chi_k^2}\left(k\delta\,2^{2mj/k}\right)
\eqqcolon\alpha_V.
\label{eq:list-comparison-P-integral}
\end{align}

Next evaluate the integral under $Q=\lambda\otimes P_{\overline Y^n}$.
For a given $y^n$, the accepted source set is empty if any unerased
block disagrees with $x^n$; otherwise, it is the ball of radius
$r_{j(y^n)}$. Conditional on an erasure pattern of positive
probability with $j$ unerased blocks, agreement has
probability $2^{-mj}$, while the
accepted ball has volume $2^{mj}v_{k,d}$. Integrating over the source
and grouping the output sequences by the number of unerased blocks
therefore gives
\begin{align}
\int T_{x^n}\,dQ
&=\sum_{y^n\in\mathcal Y^n}P_{\overline Y^n}(y^n)
\int_{\mathbb R^k}T_{x^n}(s^k,y^n)\,d\lambda(s^k)\notag\\
&=\sum_{j=0}^n\pi_j\,2^{-mj}
\left(2^{mj}v_{k,d}\right)\notag\\
&=v_{k,d}.
\label{eq:list-comparison-Q-integral}
\end{align}

\textit{The implication between the two converse conditions:}
Suppose now that condition \eqref{m-} holds. Then
\begin{align}
\epsilon
&\geq\sum_{j=0}^n\pi_j\Theta_j(k,d)\notag\\
&\geq\sum_{j=0}^n\pi_j
\left[1-F_{\chi_k^2}\left(k\delta\,2^{2mj/k}\right)\right]\notag\\
&=1-\alpha_V.
\label{eq:list-comparison-implication}
\end{align}
Consequently, \eqref{eq:list-comparison-P-integral} gives
$\int T_{x^n}\,dP_{x^n}=\alpha_V\geq1-\epsilon$, so the test is
admissible in the infimum defining $\beta_{1-\epsilon}(P_{x^n},Q)$.
Its $Q$-integral equals $v_{k,d}$ by
\eqref{eq:list-comparison-Q-integral}. It follows that
\begin{align*}
\beta_{1-\epsilon}(P_{x^n},Q)
\leq\int T_{x^n}\,dQ
=v_{k,d},
\end{align*}
which is \eqref{27rvbf}. Therefore, every parameter choice satisfying
\eqref{m-} also satisfies the Kostina--Verd\'u specialization with
$Q_S=\lambda$ and the auxiliary output law
\eqref{eq:list-comparison-auxiliary}. Equivalently, any parameter
choice ruled out by this specialization is also ruled out by
\eqref{m-}.

\section{Proof of Theorem \ref{thm_semiasympconverse} \label{thm_semiasympconverse_proof}}

Let $h=\log(2^m)$ and define
\begin{align}
J_n&=h\sum_{i=1}^n B_i, &
\mu_n&=h\sum_{i=1}^n(1-p_i),\notag\\
V_n&=h^2\sum_{i=1}^n p_i(1-p_i), &
S_{n,k}&=\sqrt{k/2+V_n},
\label{eq:converse-mean-variance}
\end{align}
where $B_1,\ldots,B_n$ are independent and $B_i\sim\operatorname{Bern}(1-p_i)$. Let $\Upsilon=\|S^k\|^2/(k\sigma^2)$ be independent of $J_n$, and write
\begin{align}
A_k(\upsilon)&=-\log \alpha_{k,d}(\upsilon), &
b_k&=kR(d)+\frac12\log k, &
s_k&=\sqrt{k/2},
\label{eq:converse-source-variable}
\end{align}
where $\alpha_{k,d}(\upsilon)$ is defined in Theorem \ref{convo2}. Throughout the proof, unspecified constants may change from one occurrence to the next. Before the final inversion, these constants depend only on $\delta$, and also on $m$ when the channel variance is used. In particular, they are independent of $n$ and $p_1,\ldots,p_n$.

\textit{Normal approximation of the source term:} We apply Proposition~\ref{prop:normal-approximation} with $f_k=A_k$ and the centering $b_k$ defined in \eqref{eq:converse-source-variable}. Fix $0<\eta<(1-\delta)/2$ as a function of $\delta$, and let $I=[1-\eta,1+\eta]$. For $\upsilon\in I$, we have $\upsilon>\delta$ and hence 
\begin{equation}
    \alpha_{k, d}(\upsilon) = L_k\left(\sqrt{1-\frac{\delta}{\upsilon}}\right). 
\end{equation}
Then the beta-integral representation \eqref{Lin546tegral} of $L_k(\cdot)$ gives, for $k\geq2$,
\begin{align}
\alpha_{k,d}(\upsilon)
&=\lambda_k\left(\frac{\delta}{\upsilon}\right)^{(k-1)/2}
\ell_k\left(\frac{\delta}{\upsilon}\right),
\label{eq:maximal-cap-factorization}
\end{align}
where $\lambda_k$ and $\ell_k$ are defined in \eqref{lamlamlam} and \eqref{ellnndef}. Here $\delta/\upsilon$ lies in a fixed compact subinterval of $(0,1)$, and
\begin{align}
1\leq\ell_k\left(\frac{\delta}{\upsilon}\right)
&\leq\left(1-\frac{\delta}{1-\eta}\right)^{-1/2}.
\end{align}
Stirling's formula yields
\begin{align}
-\log\lambda_k=\frac12\log k+\frac12\log(2\pi)+O(k^{-1}).
\end{align}
Consequently, for all sufficiently large $k$ and $\upsilon \in I$,
\begin{align}
A_k(\upsilon)&=b_k+\frac{k}{2}\log\upsilon+r_k(\upsilon),\qquad
\sup_{\upsilon\in I}|r_k(\upsilon)|\leq C,
\label{eq:maximal-cap-log-expansion}
\end{align}
where $C$ depends only on $\delta$. Taylor's theorem gives $|\log\upsilon-(\upsilon-1)|\leq C(\upsilon-1)^2$ on $I$. Combining this bound with \eqref{eq:maximal-cap-log-expansion}, we obtain
\begin{align*}
\left|A_k(\upsilon)-b_k-\frac{k}{2}(\upsilon-1)\right|
&\leq C\bigl(1+k(\upsilon-1)^2\bigr),
\qquad \upsilon\in I.
\end{align*}
Thus $A_k$ satisfies the local condition \eqref{eq:local-quadratic-expansion}. Since $\Upsilon\stackrel{d}{=}\chi_k^2/k$, Proposition~\ref{prop:normal-approximation} gives constants $K$ and $k_1$, depending only on $\delta$, such that for $k\geq k_1$,
\begin{align}
\sup_{w\in\mathbb R}
\left|\mathbb P\left[\frac{A_k(\Upsilon)-b_k}{s_k}\leq w\right]-\Phi(w)\right|
&\leq\frac{K}{\sqrt{k}}.
\label{eq:maximal-cap-normal}
\end{align}

\textit{An exponential random-threshold representation:} Let $E$ have the exponential distribution with mean one, independently of $(\Upsilon,J_n)$. For every real $x$,
\begin{align}
\max\{0,1-e^{-x}\}=\mathbb P[E<x]. \label{exportrep}
\end{align}
The spherical-cap component of Theorem \ref{convo2} can therefore be written exactly as
\begin{align}
\mathcal B_{n,k}
&\coloneqq\sum_{j=0}^n\pi_j
\int_0^\infty p_{\Upsilon}(\upsilon)\max\{0,1-2^{mj}\alpha_{k,d}(\upsilon)\}\,d\upsilon\notag\\
&=\mathbb P[A_k(\Upsilon)-J_n-E>0].
\label{eq:exponential-converse-representation}
\end{align}
By Theorem \ref{convo2}, every $(d,\epsilon,n,k)$ code satisfies $\epsilon\geq\mathcal B_{n,k}$.

Conditioning on $J_n=j$ and $E=e$, and using the independence of $\Upsilon$ from $(J_n,E)$, the bound \eqref{eq:maximal-cap-normal} gives, for every $j,e$,
\begin{align*}
\left|
\mathbb P\left[A_k(\Upsilon)>j+e\right]
-
Q\left(\frac{j+e-b_k}{s_k}\right)
\right|
\leq \frac{K}{\sqrt{k}}.
\end{align*}
Moreover, by conditioning on $(J_n,E)$ in \eqref{eq:exponential-converse-representation},
\begin{align*}
\mathcal B_{n,k}
&=
\mathbb E\left[
\mathbb P\left[
A_k(\Upsilon)>J_n+E
\,\middle|\,
J_n,E
\right]
\right].
\end{align*}
Hence, applying the preceding pointwise bound and using $|\mathbb E[X]|\leq \mathbb E[|X|]$,
\begin{align*}
\left|
\mathcal B_{n,k}
-
\mathbb E\left[
Q\left(\frac{J_n+E-b_k}{s_k}\right)
\right]
\right|
&\leq
\mathbb E\left[
\left|
\mathbb P\left[
A_k(\Upsilon)>J_n+E
\,\middle|\,
J_n,E
\right]
-
Q\left(\frac{J_n+E-b_k}{s_k}\right)
\right|
\right] \\
&\leq \frac{K}{\sqrt{k}}.
\end{align*}
Since $\mathbb E[E]=1$ and $\sup_x|Q'(x)|=1/\sqrt{2\pi}$, removing $E$ changes this expectation by at most $1/(s_k\sqrt{2\pi})$. Thus
\begin{align}
\left|\mathcal B_{n,k}-\sum_{j=0}^n\pi_jQ\left(\frac{hj-b_k}{s_k}\right)\right|
&\leq\frac{K'}{\sqrt{k}}.
\label{eq:converse-poisson-binomial-normal}
\end{align}
The constants in this bound depend only on $\delta$; the distribution of the channel term has been retained exactly.

\textit{Source-induced Gaussian smoothing:} Let $Z \sim\mathcal N(0,1)$ be independent of $J_n$. Then 
\begin{align}
\sum_{j=0}^n\pi_jQ\left(\frac{hj-b_k}{s_k}\right) = \mathbb P[s_k Z-(J_n-\mu_n)>\mu_n-b_k].
\end{align}
As in Appendix \ref{asymp_thm_achievability_proof}, specifically \eqref{eq:Pe-centered}, we split $s_kZ$ into $r$ independent zero-mean Gaussian summands of variance $k/(2r)$, and apply Berry--Esseen to these summands and the centered Bernoulli channel summands. The distribution and total variance are independent of $r$, while the sum of third absolute moments of the Gaussian summands tends to zero as $r\to\infty$. Also,
\begin{align}
\sum_{i=1}^n\mathbb E\left[|h(B_i-(1-p_i))|^3\right]
&\leq h^3\sum_{i=1}^n p_i(1-p_i)=hV_n.
\end{align}
Taking $r\to\infty$ in the Berry--Esseen inequality gives
\begin{align}
\sup_{t\in\mathbb R}\left|
\mathbb P[s_kZ-(J_n-\mu_n)>t]-Q\left(\frac{t}{S_{n,k}}\right)
\right|
&\leq C_{\mathrm{BE}}\frac{hV_n}{S_{n,k}^3}.
\label{eq:converse-smoothed-berry-esseen}
\end{align}
The right-hand side is uniformly $O(k^{-1/2})$ because
\begin{align}
\sup_{v\geq0}\frac{v}{(k/2+v)^{3/2}}
&=\frac{2\sqrt{2}}{3\sqrt{3}}\frac{1}{\sqrt{k}}.
\end{align}
Combining this with $(\ref{eq:converse-poisson-binomial-normal})$, we obtain constants $D$ and $k_2$, depending only on $\delta$ and $m$, such that for all $k\geq k_2$,
\begin{align}
\left| \mathcal B_{n,k} -Q\left(\frac{\mu_n- b_k}{S_{n,k}}\right)\right|
&\leq\frac{D}{\sqrt{k}}.
\label{eq:cap-converse-normal-error}
\end{align}
In particular, every $(d,\epsilon,n,k)$ code satisfies
\begin{align}
\epsilon
&\geq Q\left(\frac{\mu_n-kR(d)-\frac12\log k}{S_{n,k}}\right)-\frac{D}{\sqrt{k}}.
\label{eq:converse-normal-probability}
\end{align}

\textit{Uniform inversion at the target probability:} Let $q=Q^{-1}(\epsilon)$. Since $p_i(1-p_i)\leq1-p_i$,
\begin{align}
V_n\leq h\mu_n,\qquad S_{n,k}^2\leq k/2+h\mu_n.
\label{eq:converse-variance-mean}
\end{align}
First suppose $\mu_n\geq2kR(d)+2k$. Then
\begin{align}
\mu_n-kR(d)-qS_{n,k}-\frac12\log k
&\geq\mu_n-kR(d)-|q|\left(\sqrt{k/2}+\sqrt{h\mu_n}\right)-\frac12\log k\notag\\
&\stackrel{(a)}{\geq} \frac{\mu_n}{2}-kR(d)-|q|\sqrt{k/2}-\frac{hq^2}{2}-\frac12\log k\notag\\
&\geq k-|q|\sqrt{k/2}-\frac{hq^2}{2}-\frac12\log k.
\label{eq:converse-large-mean}
\end{align}
In inequality $(a)$, we used $|q|\sqrt{h\mu_n}\leq\mu_n/2+hq^2/2$. The final expression is nonnegative for all sufficiently large $k$, which proves $(\ref{eq:normal-approximation-converse})$ in this case.

It remains to consider $\mu_n<2kR(d)+2k$. By $(\ref{eq:converse-variance-mean})$,
\begin{align}
S_{n,k}\leq D_0\sqrt{k},\qquad
D_0=\sqrt{\frac12+2h(R(d)+1)}.
\label{eq:converse-standard-deviation-bound}
\end{align}
Let
\begin{align}
x=\frac{\mu_n-kR(d)-\frac12\log k}{S_{n,k}}.
\end{align}
Then equation $(\ref{eq:converse-normal-probability})$ implies $Q(x)\leq\epsilon+D/\sqrt{k}$. For sufficiently large $k$, $Q(q-1) > \epsilon+D/\sqrt{k}$, so $x>q-1$. If $x\geq q$, the desired result follows immediately. 

If $q-1<x<q$, then equation $(\ref{eq:converse-normal-probability})$ implies
\begin{align}
\frac{D}{\sqrt{k}}\geq Q(x)-Q(q)
=\int_x^q\phi(u)\,du\geq\phi_*(q-x), \label{9nurl8z}
\end{align}
where $\phi$ denotes the standard Gaussian PDF and we let
\begin{align}
\phi_* =\min_{u\in[q-1,q]}\phi(u)>0.
\end{align}
Hence, \eqref{9nurl8z} gives us 
\begin{align}
\mu_n-kR(d)
&\geq S_{n,k}q+\frac12\log k-\frac{D S_{n,k}}{\phi_*\sqrt{k}}\notag\\
&\geq S_{n,k}q+\frac12\log k-\frac{DD_0}{\phi_*},
\end{align}
where the last inequality above follows from \eqref{eq:converse-standard-deviation-bound}. 
Choosing $C\geq DD_0/\phi_*$ and increasing $k_0$ to meet the finitely many lower bounds on $k$ above proves the theorem. Both $C$ and $k_0$ depend only on $\delta$, $\epsilon$ and $m$.

\section{Proofs of Propositions~\ref{prop:sscc-source-achievability}--\ref{prop:sscc-achievability}}
\label{app:sscc-achievability}

\textit{Source code (Proposition~\ref{prop:sscc-source-achievability}):}
Fix $s^k\ne\mathbf0$ and write
$\upsilon=\|s^k\|^2/(k\sigma^2)$.
By Lemma~\ref{lemmaball}, each independently generated reproduction
belongs to $B_d(s^k)$ with probability $\rho_{k,d}(\upsilon)$. Hence,
\begin{align}
 &\Pr\!\left(
 \|S^k-\operatorname{g}_{\mathrm s}
       (\operatorname{f}_{\mathrm s}(S^k))\|^2>kd
 \,\middle|\,S^k=s^k\right) = \bigl(1-\rho_{k,d}(\upsilon)\bigr)^M.
 \label{eq:sscc-source-conditional-error}
\end{align}
Averaging over $\Upsilon$ gives
\eqref{eq:sscc-source-achievability}. To prove~\eqref{eq:sscc-one-ball}, note that rotational invariance
of $S^k$ makes $\Pr(\|S^k-z^k\|^2\leq kd)$ constant over all
$z^k$ with $\|z^k\|^2=k(\sigma^2-d)$.
Averaging this constant over an independent
$\bar Z^k\sim P_{Z^k}$ and then conditioning on the source gives
\begin{align}
 \Pr(\|S^k-z^k\|^2\leq kd)
 &=\Pr(\|S^k-\bar Z^k\|^2\leq kd)\notag\\
 &=\int_0^\infty p_{\Upsilon}(\upsilon)
                  \rho_{k,d}(\upsilon)\,d\upsilon
 =b_{\mathrm s}(k,d).
 \label{eq:sscc-fixed-ball-mass}
\end{align}
This proves Proposition~\ref{prop:sscc-source-achievability}.

\textit{Channel code (Proposition~\ref{prop:sscc-channel-achievability}):}
Fix an erasure pattern of positive probability with unerased positions
$\mathcal J\subseteq\{1,\ldots,n\}$, where $|\mathcal J|=j$. Let $E_{\mathcal J}$ denote the event that the channel output $Y^n$ has this erasure pattern. Suppose further that message $t\in[M]$ is transmitted.
The transmitted codeword $\bar X_t^n$ always agrees with the output
in every unerased position. Conditional on $\bar X_t^n$, each
preceding codeword agrees with it in all $j$ unerased positions
with probability $2^{-mj}$, independently of the other preceding
codewords. The decoder selects $t$ precisely when none of these
$t-1$ codewords agrees. Therefore, averaging over the random channel
codebook,
\begin{align}
 \Pr(\widehat U=t\mid U=t,E_{\mathcal J})
 &=(1-2^{-mj})^{t-1}.
 \label{eq:sscc-channel-conditional-success}
\end{align}
For $j=0$, the decoder always outputs $1$, so the same expression
holds with the convention $0^0=1$.
Averaging over the uniform message gives
\begin{align}
 \Pr(\widehat U=U\mid E_{\mathcal J})
 &=\frac1M\sum_{t=1}^M(1-2^{-mj})^{t-1}
 =\frac{1-(1-2^{-mj})^M}{M2^{-mj}}.
 \label{eq:sscc-channel-pattern-success}
\end{align}
The right-hand side depends on the erasure pattern only through $j$.
Averaging over the number of unerased blocks, whose PMF is
$\{\pi_j\}_{j=0}^n$, yields
$\Pr(\widehat U=U)=H_{\mathrm A}(M;n)$.
The term corresponding to $t=1$ in the preceding sum is $1$, while
all its terms lie in $[0,1]$. Thus
$1/M\leq H_{\mathrm A}(M;n)\leq1$.
Finally, some deterministic channel codebook has uniform-message
correct-decoding probability at least its ensemble average
$H_{\mathrm A}(M;n)$. This proves
Proposition~\ref{prop:sscc-channel-achievability}.

\textit{Symmetrized SSCC (Proposition~\ref{prop:sscc-achievability}):}
First suppose $M\geq2$.
For a realization
$\mathcal B_{\mathrm s}=(z_1^k,\ldots,z_M^k)$ of the source codebook,
define
\begin{align}
 C(s^k,\mathcal B_{\mathrm s})
 &\coloneqq\sum_{i=1}^M
       \mathds{1}\{z_i^k\in B_d(s^k)\},
 \label{eq:sscc-covering-count}\\
 a(\mathcal B_{\mathrm s})
 &\coloneqq\Pr\bigl(C(S^k,\mathcal B_{\mathrm s})=0\bigr).
 \label{eq:sscc-source-code-error}
\end{align}
Note that $a(\mathcal B_{\mathrm s})$ is exactly the
excess-distortion probability of this deterministic source code.
For a realization $\mathcal B_{\mathrm c}$ of the channel codebook,
let $W(v\mid u)$ be the probability that the channel decoder outputs
$v$ when message $u$ is transmitted, and write
\begin{align}
 H(\mathcal B_{\mathrm c})
 &\coloneqq\frac1M\sum_{u=1}^M W(u\mid u),
 \label{eq:sscc-fixed-channel-success}
\end{align}
so that $1 - H(\mathcal B_{\mathrm c})$ is exactly the average error probability over the uniformly distributed message of this deterministic channel code.

Propositions~\ref{prop:sscc-source-achievability}
and~\ref{prop:sscc-channel-achievability} give us that 
\begin{align}
 \mathbb E[a(\mathbf B_{\mathrm s})]
 &=A_{\mathrm s}(M;k,d),
 \label{eq:sscc-source-ensemble}\\
 \mathbb E[H(\mathbf B_{\mathrm c})]
 &=H_{\mathrm A}(M;n).
 \label{eq:sscc-channel-ensemble}
\end{align}
Note that the random codebook $\mathbf B_{\mathrm s}$ has i.i.d. codewords according to the spherical law in \eqref{pzkchoice}. Thus all source codebook realizations $\mathcal B_{\mathrm s}$ are such that
\begin{align}
    \|z_i^k \|^2 = k(\sigma^2 - d) \quad \text{ for all } \quad  i \in [M]. \label{mmm33300}
\end{align}
Therefore, the single-reproduction identity
in~\eqref{eq:sscc-one-ball} gives, for every such
source codebook realization,
\begin{align}
 \mathbb E[C(S^k,\mathcal B_{\mathrm s})]
 &=\sum_{i=1}^M\Pr(\|S^k-z_i^k\|^2\leq kd)
 =M b_{\mathrm s}(k,d).
 \label{eq:sscc-mean-covering-count}
\end{align}
We now combine these source coding and channel coding results by analyzing the permutation.

Fix source and channel codebooks $\mathcal{B}_{\mathrm s}$ and $\mathcal{B}_{\mathrm c}$, where $\mathcal{B}_{\mathrm s}$ satisfies \eqref{mmm33300}. Let
$I=\operatorname{f}_{\mathrm s}(S^k)$ and
$\widehat I=\Pi^{-1}(\operatorname{g}_{\mathrm c}(Y^n))$.
For each input index $i$, averaging over $\Pi$ gives the effective
transition probabilities
\begin{align}
 \Pr(\widehat I=i\mid I=i)
 &=\frac1M\sum_{u=1}^M W(u\mid u)=H(\mathcal B_{\mathrm c}),
 \label{eq:sscc-diagonal-symmetry}\\
 \Pr(\widehat I=\ell\mid I=i)
 &=\frac{1}{M(M-1)}
       \sum_{u=1}^M\sum_{\substack{v\in[M]\\v\ne u}}W(v\mid u)
 =\frac{1-H(\mathcal B_{\mathrm c})}{M-1},\qquad\ell\ne i.
 \label{eq:sscc-off-diagonal-symmetry}
\end{align}
The first equality uses the uniformity of $\Pi(i)$ over the set $[M]$.
The second uses the uniformity of the ordered pair
$(\Pi(i),\Pi(\ell))$ over all distinct labels.
These identities also hold conditional on any fixed source
realization with selected index $i$, since the permutation is
independent of the source and component codes.

Fix $S^k=s^k$ and set $c=C(s^k,\mathcal B_{\mathrm s})$.
If $c=0$, no decoder output can meet the distortion constraint.
If $c>0$, the source encoder selects one covering index, and the
remaining $c-1$ covering indices each have probability
$$\frac{1-H(\mathcal B_{\mathrm c})}{M-1}$$ 
of being selected at the receiver. Hence,
\begin{align}
 \Pr(\|s^k-Z^k\|^2\leq kd\mid S^k=s^k)
 &=\begin{cases}
 H(\mathcal B_{\mathrm c})+\dfrac{1-H(\mathcal B_{\mathrm c})}{M-1}(c-1),&c>0,\\[4pt]
 0,&c=0.
 \end{cases}
 \label{eq:sscc-conditional-distortion-success}
\end{align}
All probabilities here are conditional on the fixed component
codebooks and average over the erasures and permutation. For brevity, write $C=C(S^k,\mathcal B_{\mathrm s})$.
Then using
\eqref{eq:sscc-source-code-error} and
\eqref{eq:sscc-mean-covering-count}, we obtain 
\begin{align}
\Pr(\|S^k-Z^k\|^2\le kd)
&=
\sum_{c=1}^M
\left(
H(\mathcal B_{\mathrm c})+\frac{1-H(\mathcal B_{\mathrm c})}{M-1}(c-1)
\right)
\Pr(C=c)\\
&=
H(\mathcal B_{\mathrm c})\sum_{c=1}^M\Pr(C=c)
+
\frac{1-H(\mathcal B_{\mathrm c})}{M-1}
\sum_{c=1}^M(c-1)\Pr(C=c)\\
 &= H(\mathcal B_{\mathrm c})\bigl(1-a(\mathcal B_{\mathrm s})\bigr) +\frac{1-H(\mathcal B_{\mathrm c})}{M-1}
 \left(Mb_{\mathrm s}(k,d)-\bigl(1-a(\mathcal B_{\mathrm s})\bigr)\right).
 \label{eq:sscc-fixed-code-success}
\end{align}
Equivalently,
\begin{align}
 \Pr(\|S^k-Z^k\|^2>kd)
 &=\eta\left(H(\mathcal B_{\mathrm c})\right)\, a(\mathcal B_{\mathrm s})
   +\left(1-\eta(H(\mathcal B_{\mathrm c}))\right)\, \bigl(1-b_{\mathrm s}(k,d)\bigr),
 \label{eq:sscc-exact-distortion-identity}\\
 \eta(H)&\coloneqq\frac{MH-1}{M-1}.
 \label{eq:sscc-eta-general}
\end{align}
Now we average over the independently generated source and channel codebooks.
The random variables $a(\mathbf B_{\mathrm s})$ and
$H(\mathbf B_{\mathrm c})$ are independent, and $\eta$ is affine.
Thus~\eqref{eq:sscc-source-ensemble}--\eqref{eq:sscc-channel-ensemble}
imply
\begin{align}
 &\Pr(\|S^k-Z^k\|^2>kd)\notag\\
 &\quad=\mathbb E\!\left[
       \eta(H(\mathbf B_{\mathrm c}))a(\mathbf B_{\mathrm s})\right]
    +\left(1-\mathbb E[\eta(H(\mathbf B_{\mathrm c}))]\right)
       \bigl(1-b_{\mathrm s}(k,d)\bigr)\notag\\
 &\quad=\eta_{\mathrm A}(M;n)A_{\mathrm s}(M;k,d)
    +\bigl(1-\eta_{\mathrm A}(M;n)\bigr)
       \bigl(1-b_{\mathrm s}(k,d)\bigr)\notag\\
 &\quad=A_{\mathrm{sep}}(M;n,k,d).
 \label{eq:sscc-ensemble-calculation}
\end{align}
This proves the exact ensemble-average identity for $M\geq2$.
Some pair of deterministic source and channel codebooks has
permutation-averaged excess-distortion probability no larger than
this ensemble average. Keeping the fresh shared permutation for
that pair proves the existence claim.

For completeness, we show that the product bound in \eqref{eq:sscc-product-achievability} follows from \eqref{eq:sscc-achievability-condition}. Indeed, 
\begin{align}
 &1-\bigl(1-A_{\mathrm s}(M;k,d)\bigr)H_{\mathrm A}(M;n)
       -A_{\mathrm{sep}}(M;n,k,d)\notag\\
 &\qquad=\frac{1-H_{\mathrm A}(M;n)}{M-1}
 \left(Mb_{\mathrm s}(k,d)
             -\bigl(1-A_{\mathrm s}(M;k,d)\bigr)\right)
 \geq0.
 \label{eq:sscc-product-gap}
\end{align}
The inequality follows by integrating
$1-(1-\rho)^M\leq M\rho$, valid for $0\leq\rho\leq1$.
This proves the weaker sufficient condition
in~\eqref{eq:sscc-product-achievability}.

Finally, when $M=1$, the specified SSCC construction uses the
constant reproduction $\mathbf0$. Therefore,
\begin{align*}
 \Pr(\|S^k-Z^k\|^2>kd)
 &=\Pr(\|S^k\|^2>kd)
 =1-F_{\chi_k^2}(k\delta)
 =A_{\mathrm{sep}}(1;n,k,d).
\end{align*}
This proves the $M=1$ case. Finally, we show that the minimum
in~\eqref{eq:sscc-achievability-condition} is attained.
Indeed, $M H_{\mathrm A}(M;n)\leq 2^{mn}$, so
$\eta_{\mathrm A}(M;n)\to0$ as $M\to\infty$, and hence
\begin{align*}
 \lim_{M\to\infty}A_{\mathrm{sep}}(M;n,k,d)
 &=1-b_{\mathrm s}(k,d)\\
 &>1-F_{\chi_k^2}(k\delta)
 =A_{\mathrm{sep}}(1;n,k,d).
\end{align*}
The strict inequality holds because the spherical reproductions
are nonzero, while a centered distortion ball has strictly greater
Gaussian probability than any nonzero translate.
Thus all sufficiently large interface sizes give a value larger
than the $M=1$ choice, and the minimum is attained at a finite $M$.
This completes the proof.

\bibliographystyle{IEEEtran}

\bibliography{citations}

@misc{adeel_unknown_channels,
      title={Lossy Joint Source-Channel Coding over Unknown Channels}, 
      author={Adeel Mahmood and Harish Viswanathan and Jinfeng Du},
      year={2026},
      eprint={2606.07933},
      archivePrefix={arXiv},
      primaryClass={cs.IT},
      url={https://arxiv.org/abs/2606.07933}, 
}

@misc{nargis,
      title={Towards Robust Semantic Video Transmission over Block Erasure Channels}, 
      author={Nargis Fayaz and Homa Esfahanizadeh and Matin Mortaheb and Jinfeng Du and Harish Viswanathan},
      year={2026},
      eprint={2607.07823},
      archivePrefix={arXiv},
      primaryClass={eess.IV},
      url={https://arxiv.org/abs/2607.07823}, 
}

@ARTICLE{lancho2020saddlepoint,
  author={Lancho, Alejandro and Östman, Johan and Durisi, Giuseppe and Koch, Tobias and Vazquez-Vilar, Gonzalo},
  journal={IEEE Transactions on Wireless Communications}, 
  title={Saddlepoint Approximations for Short-Packet Wireless Communications}, 
  year={2020},
  volume={19},
  number={7},
  pages={4831-4846},
  doi={10.1109/TWC.2020.2987573}}

@ARTICLE{durisi2016toward,
  author={Durisi, Giuseppe and Koch, Tobias and Popovski, Petar},
  journal={Proceedings of the IEEE}, 
  title={Toward Massive, Ultrareliable, and Low-Latency Wireless Communication With Short Packets}, 
  year={2016},
  volume={104},
  number={9},
  pages={1711-1726},
  doi={10.1109/JPROC.2016.2537298}}

@book{NISTHandbook,
  title={NIST Handbook of Mathematical Functions},
  editor={Olver, Frank W.~J. and Lozier, Daniel W. and Boisvert, Ronald F. and Clark, Charles W.},
  year={2010},
  publisher={Cambridge University Press},
  address={New York},
  isbn={9780521140638}
}

@ARTICLE{kostina_SC,
  author={Kostina, Victoria and Verd{\'u}, Sergio},
  journal={IEEE Transactions on Information Theory}, 
  title={Fixed-Length Lossy Compression in the Finite Blocklength Regime}, 
  year={2012},
  volume={58},
  number={6},
  pages={3309-3338},
  doi={10.1109/TIT.2012.2186786}}

@INPROCEEDINGS{cheng2020,
  author={Cheng, Zhengxue and Sun, Heming and Takeuchi, Masaru and Katto, Jiro},
  booktitle={2020 IEEE/CVF Conference on Computer Vision and Pattern Recognition (CVPR)}, 
  title={Learned Image Compression With Discretized {Gaussian} Mixture Likelihoods and Attention Modules}, 
  year={2020},
  volume={},
  number={},
  pages={7936-7945},
  doi={10.1109/CVPR42600.2020.00796}}

@ARTICLE{6802432,
  author={Yang, Wei and Durisi, Giuseppe and Koch, Tobias and Polyanskiy, Yury},
  journal={IEEE Transactions on Information Theory}, 
  title={Quasi-Static Multiple-Antenna Fading Channels at Finite Blocklength}, 
  year={2014},
  volume={60},
  number={7},
  pages={4232-4265},
  doi={10.1109/TIT.2014.2318726}}

@ARTICLE{5452208,

  author={Polyanskiy, Yury and Poor, H. Vincent and Verd{\'u}, Sergio},

  journal={IEEE Transactions on Information Theory}, 

  title={Channel Coding Rate in the Finite Blocklength Regime}, 

  year={2010},

  volume={56},

  number={5},

  pages={2307-2359},

  doi={10.1109/TIT.2010.2043769}}

@misc{homapaper2,
      title={Block Erasure-Aware Semantic Multimedia Compression via {JSCC} Autoencoder}, 
      author={Homa Esfahanizadeh and Nargis Fayaz and Jinfeng Du and Harish Viswanathan},
      year={2026},
      eprint={2601.20707},
      archivePrefix={arXiv},
      primaryClass={cs.MM},
      url={https://arxiv.org/abs/2601.20707}, 
}

@misc{adeelUEP,
      title={Weighted Unequal Error Protection over a {R}ayleigh Fading Channel}, 
      author={Adeel Mahmood},
      year={2026},
      eprint={2602.24225},
      archivePrefix={arXiv},
      primaryClass={cs.IT},
      url={https://arxiv.org/abs/2602.24225}, 
}

@ARTICLE{alternative_oneshot_JSCC,
  author={Li, Cheuk Ting and Anantharam, Venkat},
  journal={IEEE Transactions on Information Theory}, 
  title={A Unified Framework for One-Shot Achievability via the {Poisson} Matching Lemma}, 
  year={2021},
  volume={67},
  number={5},
  pages={2624-2651},
  doi={10.1109/TIT.2021.3058842}}

@ARTICLE{kostina_JSCC,
  author={Kostina, Victoria and Verdú, Sergio},
  journal={IEEE Transactions on Information Theory}, 
  title={Lossy Joint Source-Channel Coding in the Finite Blocklength Regime}, 
  year={2013},
  volume={59},
  number={5},
  pages={2545-2575},
  doi={10.1109/TIT.2013.2238657}}

@book{Cover2006,
  author    = {T. M. Cover and J. A. Thomas},
  title     = {Elements of Information Theory},
  edition   = {2nd},
  publisher = {Wiley-Interscience},
  address   = {Hoboken, N.J.},
  year      = {2006}
}

@inproceedings{palzer_timo_converse,
  author    = {L. Palzer and R. Timo},
  title     = {A converse for lossy source coding in the finite blocklength regime},
  booktitle = {2016 International Zurich Seminar on Communications (IZS)},
  address   = {Zurich, Switzerland},
  pages     = {15--19},
  month     = mar,
  year      = {2016},
  doi       = {10.3929/ethz-a-010645199}
}

@phdthesis{palzer_thesis,
  author = {L. Palzer},
  title  = {Rate-Distortion Analysis of Sparse Sources and Compressed Sensing with Scalar Quantization},
  school = {Technische Universit{\"a}t M{\"u}nchen},
  type   = {Dr.-Ing. dissertation},
  year   = {2019},
  url    = {https://mediatum.ub.tum.de/doc/1486321/1486321.pdf}
}

\end{document}